%% file: paper.tex
\documentclass[conference]{IEEEtran}
\usepackage{times}

\usepackage[numbers]{natbib}
\usepackage{multicol}
\usepackage[hidelinks,bookmarks=true]{hyperref}

\usepackage{amsmath,amssymb,amsthm,amsfonts}
\usepackage{xspace}
\usepackage{color}
\usepackage{hyperref}
\usepackage{algorithm}
\usepackage{algorithmicx}
\usepackage[noend]{algpseudocode}
\algrenewcommand\algorithmicindent{0.75em}%

\usepackage[caption=false]{subfig}

\usepackage{multirow}
\usepackage[table,xcdraw]{xcolor}
\usepackage{graphicx}

\input{macros.tex}

\begin{document}
%
\title{Approximate bi-criteria search by efficient representation of subsets of the Pareto-optimal frontier}

\author{Boris Goldin, Oren Salzman\\
Technion---Israel Institute of Technology\\
boris.goldin@campus.technion.ac.il,
osalzman@cs.technion.ac.il
}

\maketitle

\begin{abstract}
We consider the bi-criteria shortest-path problem where we want to compute shortest paths on a graph that simultaneously balance two cost functions.
While this problem has numerous applications, there is usually no path minimizing both cost functions simultaneously. 
Thus, we typically consider the set of paths where no path is strictly better than the others in both cost functions, a set called the Pareto-optimal frontier.
Unfortunately, the size of this set may be exponential in the number of graph vertices and the general problem is \NP-hard.  
While existing schemes to approximate this set exist, they may be slower than exact approaches when applied to relatively small instances and running them on graphs with even a moderate number of nodes is often impractical.
The crux of the problem lies in how to efficiently approximate the Pareto-optimal frontier.
Our key insight is that the Pareto-optimal frontier can be approximated using \myemph{pairs} of paths.
This simple observation allows us to run a best-first search while efficiently and effectively pruning away intermediate solutions in order to obtain an approximation of the Pareto frontier for any given approximation factor.
We compared our approach with an adaptation of \boastar, the state-of-the-art algorithm for computing exact solutions to the bi-criteria shortest-path problem. 
Our experiments show that as the problem becomes harder, the speedup obtained becomes more pronounced.
Specifically, on large roadmaps, when using an approximation factor of $10\%$ we obtain a speedup on the average running time of more than~$\times 19$. 
\end{abstract}

\section{Introduction \& Related Work}
\label{sec:intro}
We consider the bi-criteria shortest-path problem, an extension to the classical (single-criteria) shortest-path problem where we are given a graph $G = (V,E)$ and each edge has two cost functions. 
Here, we are required to compute paths that balance between the two cost functions.
The well-studied problem~\cite{CP07} has numerous applications. For example, given a road network, the two cost functions can represent travel times and distances and we may need to consider the set of paths that allow to balance between these costs.
Other applications include 
planning of power-transmission lines~\cite{bachmann2018multi} 
and 
planning how to transport hazardous material in order to balance between minimizing the travel distance and the risk of exposure
for residents~\cite{bronfman2015maximin}.

There usually is no path minimizing all cost functions simultaneously. 
Thus, we typically consider the set of paths where no path is strictly better then the others for both cost functions, a set called the \myemph{Pareto-optimal frontier}.
Unfortunately, the problem is \NP-hard~\cite{S87} as the cardinality of the size of the Pareto-optimal frontier may be exponential in $|V|$~\cite{Ehrgott05,breugem2017analysis} and even
determining whether a path belongs to the Pareto-optimal frontier is \NP-hard~\cite{PY00}.

Existing methods either try to 
efficiently compute the Pareto-optimal frontier
or to 
relax the problem and only compute an approximation of this set.
\paragraph{Efficient computation of the Pareto-optimal frontier.}
To efficiently compute the Pareto-optimal frontier, adaptations of the celebrated \astar algorithm~\cite{HNR68} were suggested.
Stewart et al.~\cite{stewart1991multiobjective} introduced 
Multi-Objective A* (\algname{MOA$^*$}) which is a  multiobjective extension of $\algname{A}^*$.  
The most notable difference between \algname{MOA$^*$} and $\algname{A}^*$ is in maintaining the Pareto-optimal frontier to intermediate vertices. This requires to check if a path~$\pi$ is \myemph{dominated} by another path~$\tilde{\pi}$. Namely, if both of~$\tilde{\pi}$'s costs are smaller than~$\pi$'s costs.
As these dominance checks are repeatedly performed, the time complexity of the checks play a crucial role for the efficiency of such bi-criteria shortest-path algorithms.
\algname{MOA$^*$} was later revised~\cite{de2005new,mandow2010multiobjective,pulido2015dimensionality} with the most efficient variation, termed bi-Objective \astar (\boastar)~\cite{UYBZSK20} allowing to compute these operations in $O(1)$ time when a consistent heuristic is used.\footnote{A heuristic function is said to be consistent if its estimate is always less than or equal to the estimated distance from any neighbouring vertex to the goal, plus the cost of reaching that neighbour.}

\paragraph{Approximating the Pareto-optimal frontier.}
Initial methods in computing an approximation of the Pareto-optimal frontier were directed towards devising a Fully Polynomial Time Approximation Scheme\footnote{An FPTAS  is an approximation scheme whose time complexity is polynomial in the input size and also polynomial in $1/\eps$ where~$\eps$ is the approximation factor.} (FPTAS)~\cite{V01}.
Warburton~\cite{W87} proposed a method for finding an approximate Pareto-optimal solution to the problem for any degree of accuracy using scaling and rounding techniques. 
%
Perny and Spanjaard~\cite{perny2008near} presented another FPTAS given that a finite upper bound $L$ on the numbers of arcs of all solution-paths in the Pareto-frontier is known.
This requirement was later relaxed~\cite{TZ09,breugem2017analysis} by partitioning the space of solutions into cells according to the approximation factor and, roughly speaking, taking only one solution in each grid cell.
%
%
%
Unfortunately, the running times of FPTASs are typically polynomials of high degree, and hence they may be slower than exact approaches when applied to relatively-small instances and running them on graphs with even a moderate number of nodes (e.g., $\approx 10,000$) is often impractical~\cite{breugem2017analysis}.

A different approach to compute a subset of the Pareto-optimal frontier is to find all extreme supported non-dominated points (i.e., the extreme points on the convex hull of the Pareto-optimal set)~\cite{sedeno2015dijkstra}.
Taking a different approach Legriel et al.~\cite{LLCM10} suggest a method  based on satisfiability/constraint solvers.
Alternatively, a simple variation of \algname{MOA$^*$}, termed \algname{MOA$^*_\eps$} allows to compute an approximation of the Pareto-optimal frontier by pruning intermediate paths that are approximately dominated by already-computed solutions~\cite{perny2008near}. However, as we will see, this allows to prune only a small subset of paths that may be pruned.

Finally, recent work~\cite{BC20} conducts a comprehensive computational study with an emphasis on multiple criteria.
Similar to the aforementioned FPTASs, their framework still partitions the space prior to running the algorithm.

\paragraph{Key contribution.}
To summarize, exact methods compute a solution set whose size is often exponential in the size of the input. 
While one would expect that approximation algorithms will allow to dramatically speed up computation times, in practice their running times are often slower than exact solutions for FPTAS's because they partition the space of solutions into cells according to the approximation factor in advance. 
Alternative methods only prune paths that are approximately dominated by already-computed solutions.

Our key insight is that we can efficiently partition the space of solutions into cells during the algorithm's execution (and not a-priori). This allows us to efficiently and effectively prune away intermediate solutions in order to obtain an approximation of the Pareto-optimal frontier for any given approximation factor $\eps$ (this will be formalized in Sec.~\ref{sec:pdf}).
This is achieved by running a best-first search on \myemph{path pairs} and not individual paths. Such path pairs represent a subset of the Pareto-optimal frontier such that any solution in this subset is approximately dominated by the two paths. 
Using concepts that draw inspiration from a recent search algorithm from the robotics literature~\cite{FuKSA19}, 
we propose Path-Pair  \astar (\is).
\is dramatically reduces the computational complexity of the best-first search by merging path pairs while still ensuring that an approximation of the Pareto-optimal frontier is obtained for any desired approximation. 

For example, on a roadmap of roughly 1.5 million vertices, \is approximates the Pareto-optimal frontier within a factor of $1\%$ in roughly 13 seconds on average on a commodity laptop.
We compared our approach with an adaptation of \boastar~\cite{UYBZSK20}, the state-of-the-art algorithm for computing exact solutions to the bi-criteria shortest-path problem, which we term \boastareps. \boastareps computes near-optimal solutions by using the approach suggested in~\cite{perny2008near}.
Our experiments show that as the problem becomes harder, the speedup that \is may offer becomes more pronounced.
Specifically, on the aforementioned roadmap and using an approximation factor of $10\%$, we obtain 
a speedup on the average running time of more than $\times 19$
and a maximal speedup of over~$\times 25$.

%


\section{Problem Definition}
\label{sec:pdf}
Let $G = (V,E)$ be a graph, 
$c_1 : E \rightarrow \R$ and
$c_2 : E \rightarrow \R$ be two cost functions defined over the graph edges.
A path $\pi = v_1, \ldots v_k$ is a sequence of vertices where consecutive vertices are connected by an edge.
We extend the two cost functions to paths as follows:
$$
c_1(\pi) = \sum_{i=1}^{k-1} c_1(v_i, v_{i+1})
\hspace{2mm}
\text{ and }
\hspace{2mm}
c_2(\pi) = \sum_{i=1}^{k-1} c_2(v_i, v_{i+1}).
$$

Unless stated otherwise, 
all paths start at the same specific vertex $\vs$
and $\pi_u$ will denote a path to vertex $u$.
\begin{defn}[Dominance]
Let $\pi_u$ and $\tilde{\pi}_u$ be two paths to vertex $u$.
We say that $\pi_u$ \myemph{weakly dominates}~$\tilde{\pi}_u$ if
(i)~$c_1(\pi_u) \leq c_1(\tilde{\pi}_u)$ and
(ii)~$c_2(\pi_u) \leq c_2(\tilde{\pi}_u)$.
We say that~$\pi_u$ \myemph{strictly dominates}~$\tilde{\pi}_u$ if
(i)~$\pi_u$ {weakly dominates}~$\tilde{\pi}_u$ and
(ii)~$c_1(\pi_u) < c_1(\tilde{\pi}_u)$ or~$c_2(\pi_u) < c_2(\tilde{\pi}_u)$.
\end{defn}
\begin{defn}[Approximate dominance]
Let $\pi_u$ and $\tilde{\pi}_u$ be two paths to vertex $u$
and 
let $\eps_1 \geq 0$ and $\eps_2 \geq 0$ be two real values.
We say that $\pi_u$ \myemph{$(\eps_1,\eps_2)$-dominates}~$\tilde{\pi}_u$ if
(i)~$c_1(\pi_u) \leq (1 + \eps_1) \cdot c_1(\tilde{\pi}_u)$ and
(ii)~$c_2(\pi_u) \leq (1 + \eps_2) \cdot  c_2(\tilde{\pi}_u)$.
When $\eps_1 = \eps_2$, we will sometimes say that $\pi_u$ \myemph{$(\eps_1)$-dominates}~$\tilde{\pi}_u$ and call $\eps_1$ the \myemph{approximation factor}.
\end{defn}
\begin{defn}[(approximate) Pareto-optimal frontier]
The~\myemph{Pareto-optimal frontier} $\Pi_{u}$ of a vertex $u$ is a set of paths connecting $\vs$ and~$u$ such that 
(i)~no path in $\Pi_{u}$ is strictly dominated by any other path from $\vs$ to $u$
and
(ii)~every path from $\vs$ to $u$ is weakly dominated by  a path in $\Pi_{u}$.
Similarly, for $\eps_1 \geq 0$ and $\eps_2 \geq 0$ the \myemph{approximate Pareto-optimal frontier}\footnote{Our definition of an approximate Pareto-optimal frontier slightly differs from existing definitions~\cite{breugem2017analysis} which do not require that the  approximate Pareto frontier is a subset of the Pareto-optimal frontier.} $\Pi_{u}(\eps_1,\eps_2) \subseteq \Pi_u$ is a subset of $u$'s Pareto frontier such that every path in $\Pi_{u}$ is $(\eps_1,\eps_2)$-dominated by a path in $\Pi_{u}(\eps_1,\eps_2)$.
\end{defn}
\noindent
For brevity we will use the terms 
(approximate) Pareto frontier
to refer to the 
(approximate) Pareto-optimal frontier.
For a visualization of these notions, see Fig.~\ref{fig:dominance}.
\begin{figure}[tb]
  \centering
  \includegraphics[height=4.5cm]{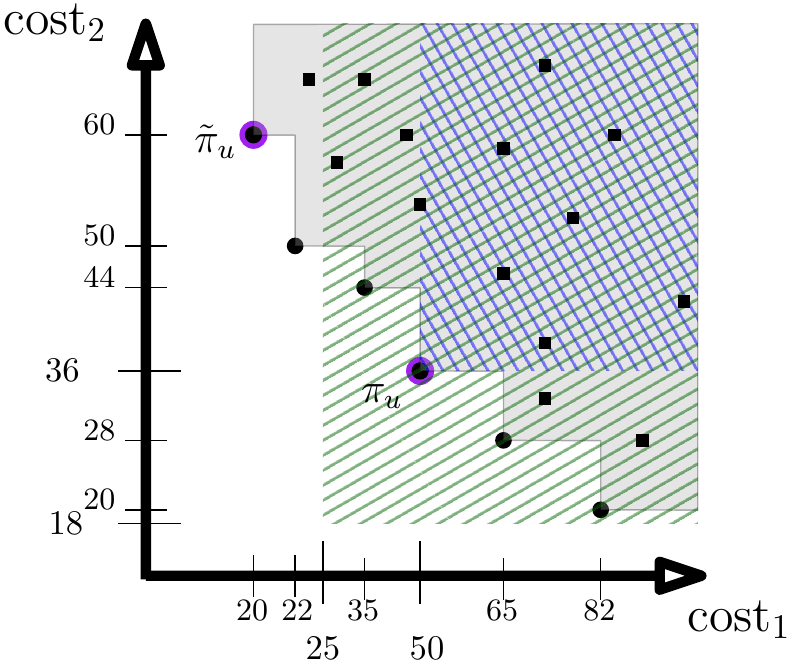}
  \caption{(approximate) Dominance and (approximate) Pareto frontier. Given start and target vertices, we consider each path $\pi_u$ as a 2D point $(c_1(\pi_u), c_2(\pi_u))$ according to the two cost functions (points and squares). 
  The set of all possible paths dominated and approximately dominated by path $\pi_u$  are depicted in blue and green, respectively (for $\varepsilon_1 = \varepsilon_2 = 1$). The Pareto frontier $\Pi_u$ is the set of all black points that collectively dominate all other possible paths (squares in grey region). 
  Finally, an approximate Pareto frontier $\Pi_u(1,1) = \{\pi_u, \tilde{\pi}_u \}$ is depicted by the two purple circles.}
  \label{fig:dominance}
\end{figure}

We are now ready to formally define our search problems.
\begin{prob}[Bi-criteria shortest path]
\label{prob:1}
Let $G$ be a graph,
$c_1, c_2 : E \rightarrow \R$ two cost functions 
and
$\vs$ and~$\vg$ be start and goal vertices, respectively.
The \myemph{bi-criteria shortest path} problem calls for computing the Pareto frontier $\Pi_{\vg}$.
\end{prob}
\begin{prob}[Bi-criteria approximate shortest path]
\label{prob:2}
Let $G$ be a graph,
$c_1, c_2 : E \rightarrow \R$ two cost functions 
and
$\vs$ and~$\vg$ be start and goal vertices, respectively.
Given~$\eps_1 \geq 0$ and $\eps_2 \geq 0$, the \myemph{bi-criteria approximate shortest path} problem calls for computing an approximate Pareto frontier~$\Pi_{\vg}(\eps_1, \eps_2)$.
\end{prob}


\section{Algorithmic Background}
\label{sec:background}
In this section we describe two approaches to solve the bi-criteria shortest-path problem (Problem~\ref{prob:1}).
With the risk of being tedious, we start with a brief review of best-first search algorithms as both state-of-the-art bi-criteria shortest path algorithms, as well as ours, rely heavily on this algorithmic framework.
We note that the description of best-first search  we present here can be optimized but this version will allow us to better explain the more advanced algorithms.

A best-first search algorithm   (Alg.~\ref{alg:astar}) computes a shortest path from $\vs$ to $\vg$ by maintaining a priority queue, called an OPEN list, that contains all the nodes that have not been expanded yet (line~\ref{alg:astar-line1}). 
Each node is associated with a path $\pi_u$ from $\vs$ to some vertex $u\in V$ (by a slight abuse of notation we will use paths and nodes interchangeability which will simplify algorithm's descriptions in the next sections). 
This queue is ordered according to some cost function called the $f$-value of the node. 
For example, in Dijkstra and \astar, this is the computed cost from~$\vs$ (also called its $g$-value) and the computed cost from~$\vs$ added to the heuristic estimate to reach~$\vg$, respectively.

At each iteration (lines~\ref{alg:astar-line3}-\ref{alg:astar-line13}), 
the algorithm extracts the most-promising node from OPEN (line~\ref{alg:astar-line3}),
checks if it has the potential to be a better solution than any found so far (line~\ref{alg:astar-line4}).
If this is the case and we reached $\vg$, the solution set is updated (in single-criteria shortest path, once a solution is found, the search can be terminated).
If not, we extend the path represented by this node to each of it's neighbors (line~\ref{alg:astar-line10}).
Again, we check if it has the potential to be a better solution than any found so far (line~\ref{alg:astar-line11}).
If this is the case, it is added to the OPEN list.

Different single-criteria search algorithms such as Dijkstra, \astar, \astareps as well as bi-criteria search algorithms such \boastar fall under this framework.
They differ with how OPEN is ordered and how the different functions (highlighted in Alg.~\ref{alg:astar}) are implemented.

\begin{algorithm}[t!]
    \textbf{Input: ($G = (V,E), \vs, \vg, \ldots$)}
    \begin{algorithmic}[1]
        
        \State {${\rm OPEN}\gets$ new node $\pi_{\vs}$}
        \label{alg:astar-line1} 

        \vspace{1mm}
        \While{${\rm{OPEN}} \neq \emptyset$}
          \State $\pi_u \gets$ \rm{OPEN{}.\func{extract\_min}}()
          \label{alg:astar-line3}
          \If {\func{is\_dominated}($\pi_u$) \label{alg:astar-line4}}
            \State {\textbf{continue}}
          \EndIf
        \vspace{1mm}
          \If {$u=\vg$} \Comment{reached goal}
            \State {\func{merge\_to\_solutions}($\pi_u$, solutions)}
            \State {\textbf{continue}}
          \EndIf
        \vspace{1mm}
          \For{$e=(u,v) \in$ {neighbors}($u, G$)}
            \State $\pi_v \gets$ \func{extend}($\pi_u, e$)
            \label{alg:astar-line10} 
            \If {\func{is\_dominated}($\pi_v $) \label{alg:astar-line11} } 
              \State {\textbf{continue}}
            \EndIf
        \vspace{1mm}
            \State {\func{insert}($\pi_v , \rm{OPEN}$)}
            \label{alg:astar-line13}
          \EndFor
        
        \EndWhile
       \State{\textbf{return} all extreme paths in solutions}
\end{algorithmic}
    \caption{Best First Search}
    \label{alg:astar}
\end{algorithm}

\paragraph*{Bi-Objective \astar (\boastar)} 
To efficiently solve Problem~\ref{prob:1},
bi-Objective \astar (\boastar) runs a best-first search.
The algorithm is endowed with two heuristic functions $h_1, h_2$ estimating the cost to reach $\vg$ from any vertex
according to~$c_1$ and $c_2$, respectively.
Here, we assume that these heuristic functions  are admissible and consistent.
This is key as the efficiency of \boastar relies on this assumption.

Given a node $\pi_u$, we define 
$g_i(\pi_u)$ to be the computed distance according to $c_i$.
It can be easily shown that in best-first search algorithms $g_i := c_i(\pi_u)$.
Additionally, we define 
$f_i(\pi_u) := g_i(\pi_u) + h_i(\pi_u)$.
Although the cost and the $g$-value of a path can be used interchangeably, we will use the former to describe general properties of paths and the latter to describe algorithm operations.
Nodes in OPEN are ordered lexicographically according to $(f_1, f_2)$  which concludes the description of how \func{extract\_min} and \func{insert} (lines~\ref{alg:astar-line3} and~\ref{alg:astar-line13}, respectively) are implemented.

Domination checks, which are typically time-consuming in bi-criteria search algorithms are implemented in $O(1)$ per node by maintaining for each vertex $u \in V$ the minimal cost to reach $u$ according to $c_2$ computed so far.
This value is maintained in a map $g_2^{\rm min}: V \rightarrow \R$ which is initialized to~$\infty$ for each vertex.
This allows to implement the function 
\func{is\_dominated} for a node $\pi_u$ by testing
if 
\begin{equation}
\label{eq:d00}
g_2(\pi_u) \geq g_2^{\rm min}(u) 
\text{ or }
f_2(\pi_u) \geq g_2^{\rm min}(\vg).
\end{equation}

The first test checks if the node is dominated by an already-extended node and replaces the CLOSED list typically used in \astar-like algorithms.
The second test checks if the node has the potential to reach the goal with a solution whose cost is not dominated by any existing solution. 
Finally, the function \func{merge\_to\_solutions} simply adds a newly-found solution to the solution set.

\paragraph{Computing the approximate Pareto frontier}
Perny and Spanjaard~\cite{perny2008near} suggest to compute an approximate Pareto frontier by endowing the algorithm with an approximation factor $\eps$.
When a node is popped from OPEN, we test if its $f$-value is $\eps$-dominated by any solution that was already computed.
While this algorithm was presented before \boastar and hence uses computationally-complex dominance checks, we can easily use this approach to adapt \boastar to compute an approximate Pareto frontier.
This is done by replacing the dominance check in Eq.~\ref{eq:d00} with the test
\begin{equation}
\label{eq:d0}
g_2(\pi_u) \geq g_2^{\rm min}(v) \text{ or }
(1 + \eps) \cdot f_2(\pi_u) \geq g_2^{\rm min}(\vg).
\end{equation}
We call this algorithm \boastareps.

\section{Algorithmic Framework}
\subsection{Preliminaries}
%

Recall that (single-criteria) shortest-path algorithms such as \astar find a solution by computing the shortest path to all nodes that have the potential to be on the shortest path to the goal (namely, whose $f$-value is less than the current estimate of the cost to reach $\vg$).
Similarly, bi-criteria search algorithms typically compute for each node the subset of the Pareto frontier that has the potential to be in~$\Pi_{\vg}$.

Now, near-optimal (single-criteria) shortest-path algorithms such as \astareps~\cite{pearl1982studies} attempt to speed up this process by only  \myemph{approximating} the shortest path to intermediate nodes.
Similarly, we suggest to construct only an approximate Pareto frontier for intermediate nodes which, in turn, will allow to dramatically reduce computation times.
Looking at Fig.~\ref{fig:dominance}, one may suggest to run an \astar-like search and if a path $\pi_u$ on the Pareto frontier $\Pi_{u}$ of  $u$ is approximately dominated by another path $\tilde{\pi}_u \in \Pi_{u}$, then discard~$\pi_u$. Unfortunately, this does not account for paths in~$\Pi_{u}$ that may have been approximately dominated by $\pi_u$ and hence discarded in previous iterations of the search. 
Existing methods use very conservative bounds to prune intermediate paths. For example, as stated in Sec.~\ref{sec:intro}, if a bound $L$ on the  length of the longest path exists, we can use this strategy by replacing $(1 + \eps)$ with $(1 + \eps)^{1/L}$ to account for error propagation~\cite{perny2008near}.

\begin{figure}[tb]
  \centering
  \includegraphics[height=4.5cm]{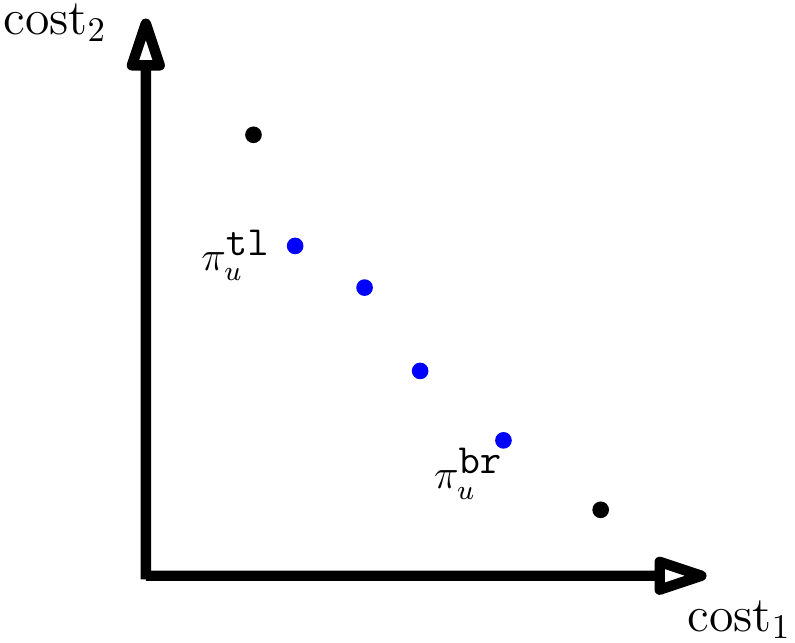}
  \caption{The partial Pareto frontier of two paths $\pi_u^{\texttt{tl}}$ and~$\pi_u^{\texttt{br}}$ is the set of all paths (blue dots) on the Pareto frontier (blue and black dots) between these paths. 
  Lemma~\ref{lem:ppf-dom} implies that any path represented by a blue dot is approximately dominated by $\pi_u^{\texttt{tl}}$ and $\pi_u^{\texttt{br}}$ for $
   \eps_1 = 
    \frac{c_1(\pi_u^{\texttt{br}}) - c_1(\pi_u^{\texttt{tl}})}{c_1(\pi_u^{\texttt{tl}})}
$
and
$
  \eps_2 = 
    \frac{c_2(\pi_u^{\texttt{tl}}) - c_2(\pi_u^{\texttt{br}})}{c_2(\pi_u^{\texttt{br}})}
$.}
  \label{fig:ppf}
\end{figure}

In contrast, we suggest a simple-yet-effective method to prune away approximately-dominated solutions using the notion of a partial Pareto frontier which we now define.
\begin{defn}[Partial Pareto frontier \ppf]
Let $\pi_u^{\texttt{tl}}, \pi_u^{\texttt{br}} \in \Pi_{u}$ be two paths on the Pareto frontier of vertex $u$ such that $c_1(\pi_u^{\texttt{tl}}) < c_1 (\pi_u^{\texttt{br}})$
(here, \texttt{tl} and \texttt{br} are shorthands for ``top left'' and ``bottom right'' for reasons which will soon be clear).
Their \myemph{partial Pareto frontier} 
$\ppf_{u}^{\pi_u^{\texttt{tl}}, \pi_u^{\texttt{br}}} \subseteq \Pi_{u}$ is a subset of a Pareto frontier such that 
if 
$\pi_u \in \Pi_{u}$ and 
$c_1(\pi_u^{\texttt{tl}}) < c_1 ({\pi}_u) < c_1 (\pi_u^{\texttt{br}})$
then
${\pi}_u  \in \ppf_{u}^{\pi_u^{\texttt{tl}}, \pi_u^{\texttt{br}}}$. The paths 
$\pi_u^{\texttt{tl}}, \pi_u^{\texttt{br}}$ are called the \myemph{extreme} paths of $\ppf_{u}^{\pi_u^{\texttt{tl}}, \pi_u^{\texttt{br}}}$
For a visualization, see Fig.~\ref{fig:ppf}.
\end{defn}

\begin{defn}[Bounded \ppf]
A {partial Pareto frontier} 
$\ppf_{u}^{\pi_u^{\texttt{tl}}, \pi_u^{\texttt{br}}} \subseteq \Pi_{u}$ is 
\myemph{$(\eps_1, \eps_2)$-bounded} 
if 
$$
  \eps_1 \geq \frac{c_1(\pi_u^{\texttt{br}}) - c_1(\pi_u^{\texttt{tl}})}{c_1(\pi_u^{\texttt{tl}})}
  \text{ and }
  \eps_2 \geq \frac{c_2(\pi_u^{\texttt{tl}}) - c_2(\pi_u^{\texttt{br}})}{c_2(\pi_u^{\texttt{br}})}.
$$
\end{defn}

\begin{lem}
\label{lem:ppf-dom}
If $\ppf_{u}^{\pi_u^{\texttt{tl}}, \pi_u^{\texttt{br}}}$ 
is an $(\eps_1, \eps_2)$-bounded partial Pareto frontier 
then any path in $\ppf_{u}^{\pi_u^{\texttt{tl}}, \pi_u^{\texttt{br}}}$ is 
$(\eps_1, \eps_2)$-dominated by both $\pi_u^{\texttt{tl}}$ and $\pi_u^{\texttt{br}}$.
\end{lem}
\begin{proof}
Let $\pi_u \in \ppf_{u}^{\pi_u^{\texttt{tl}}, \pi_u^{\texttt{br}}}$.
By definition, we have that 
$c_1(\pi_u^{\texttt{tl}}) < c_1 (\pi_u)$
and that
$\eps_1 \geq \frac{c_1(\pi_u^{\texttt{br}}) - c_1(\pi_u^{\texttt{tl}})}{c_1(\pi_u^{\texttt{tl}})}$.
Thus,
$$
c_1(\pi_u^{\texttt{br}}) \leq (1 + \eps_1) \cdot c_1(\pi_u^{\texttt{tl}})
              < (1 + \eps_1) \cdot c_1(\pi_u).
$$
As $c_2(\pi_u^{\texttt{br}}) < c_2(\pi_u)$, we have that $\pi_u^{\texttt{br}}$ approximately dominates $\pi_u$.

Similarly, 
by definition, we have that 
$c_2 (\pi_u) > c_2 (\pi_u^{\texttt{br}})$
and that
$\eps_2 \geq \frac{c_2(\pi_u^{\texttt{tl}}) - c_2(\pi_u^{\texttt{br}})}{c_2(\pi_u^{\texttt{br}})}$.
Thus,
$$
c_2(\pi_u^{\texttt{tl}}) \leq (1 + \eps_2) \cdot c_1(\pi_u^{\texttt{br}})
              < (1 + \eps_2) \cdot c_1(\pi_u).
$$
As $c_1(\pi_u^{\texttt{tl}}) < c_1(\pi_u)$, we have that $\pi_u^{\texttt{tl}}$ approximately dominates $\pi_u$.
\end{proof}

\subsection{Algorithmic description}
\begin{figure*}[t]%
\centering
  \subfloat[]{
    \label{fig:extend}
    \includegraphics[height=5cm]{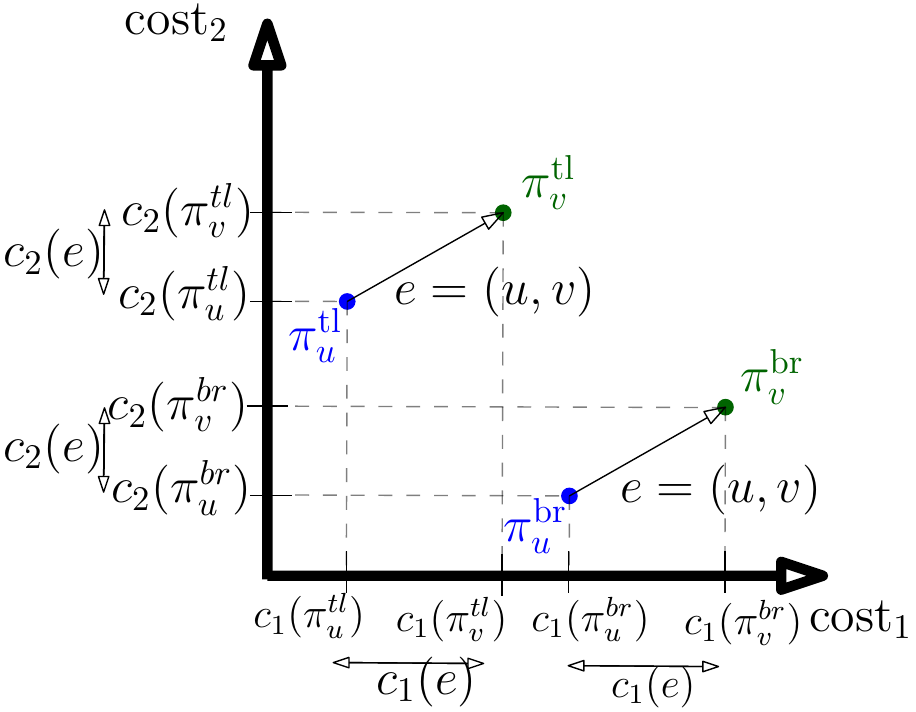}
  }
  \qquad
  \subfloat[]{
    \label{fig:merge}
    \includegraphics[height=5cm]{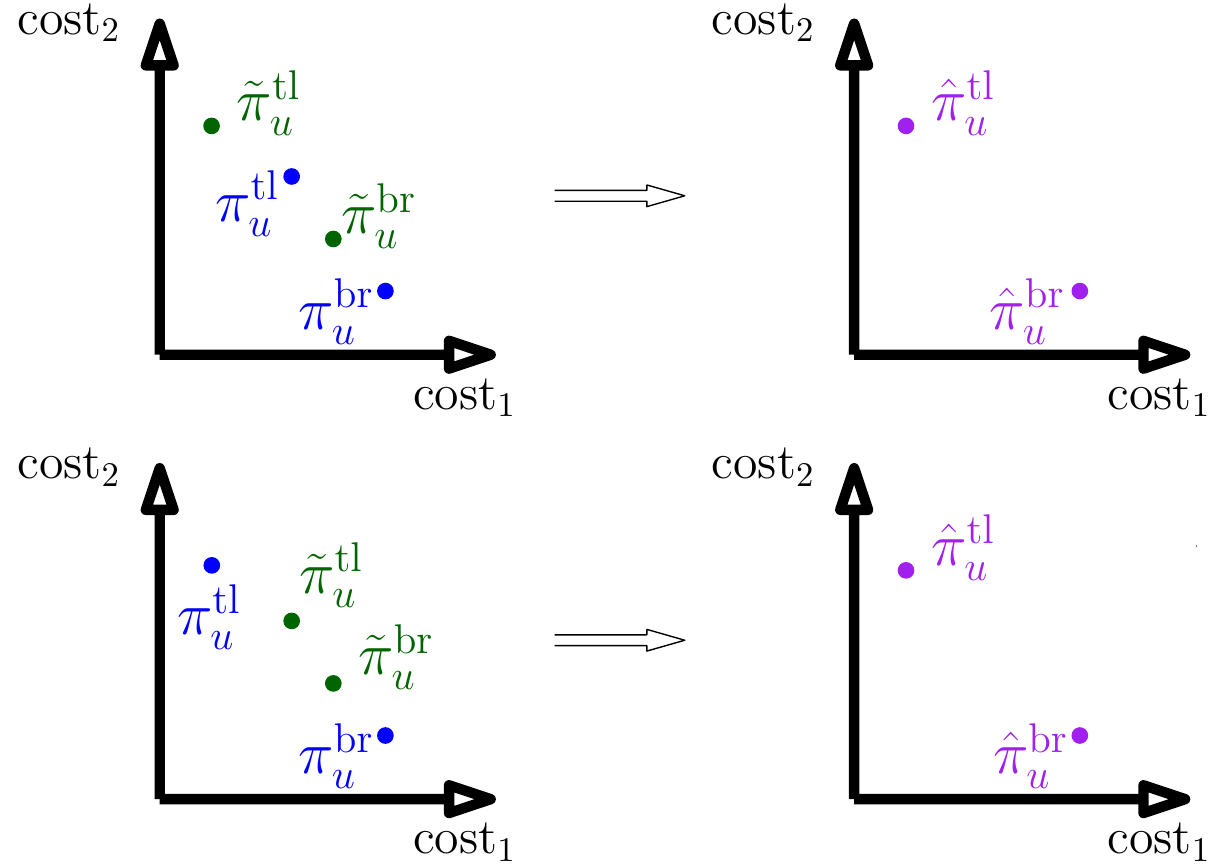}
  }  
  \caption{Operations on path pairs.
  \protect \subref{fig:extend}~Extend operation. 
  The path pair $(\pi_u^{\texttt{tl}}, \pi_u^{\texttt{br}})$ (blue) is extended by edge $e = (u,v)$ to obtain the path pair $(\pi_v^{\texttt{tl}}, \pi_v^{\texttt{br}})$ (green).  
  \protect \subref{fig:merge}~Merge operation. Two examples of merging the path pair $(\pi_u^{\texttt{tl}}, \pi_u^{\texttt{br}})$ (blue) with the path pair $(\tilde{\pi}_u^{\texttt{tl}}, \tilde{\pi}_u^{\texttt{br}})$ (green) to obtain the path pair $(\hat{\pi}_u^{\texttt{tl}}, \hat{\pi}_u^{\texttt{br}})$ (purple).
  }
\label{fig:operations}
\end{figure*} 
In contrast to standard search algorithms which incrementally construct shortest paths from $\vs$ to the graph vertices, our algorithm will incrementally construct $(\eps_1, \eps_2)$-bounded partial Pareto frontiers.
Lemma~\ref{lem:ppf-dom} suggests a method to efficiently represent and maintain these frontiers for any approximation factors $\eps_1$ and $\eps_2$.
Specifically, for a vertex $u$, \is will maintain \myemph{path pairs} corresponding to the extreme paths in partial Pareto frontiers.
For each path pair $(\pi_u^{\texttt{tl}}, \pi_u^{\texttt{br}})$ 
we have that
$c_1 (\pi_u^{\texttt{tl}}) \leq c_1(\pi_u^{\texttt{br}})$
and
$c_2 (\pi_u^{\texttt{tl}}) \geq c_2(\pi_u^{\texttt{br}})$.

Before we explain how path pairs will be used let us define operations on path pairs:
The first operation we consider is \myemph{extending} a
path pair $(\pi_u^{\texttt{tl}}, \pi_u^{\texttt{br}})$  by an edge $e = (u, v)$, which simply corresponds to extending 
both $\pi_u^{\texttt{tl}}$ and $\pi_u^{\texttt{br}}$ by~$e$.
The second operation we consider is \myemph{merging} two
path pairs $(\pi_u^{\texttt{tl}}, \pi_u^{\texttt{br}})$ and $(\tilde{\pi}_u^{\texttt{tl}}, \tilde{\pi}_u^{\texttt{br}})$.
This operation constructs a new path pair $(\hat{\pi}_u^{\texttt{tl}}, \hat{\pi}_u^{\texttt{br}})$ such that
\begin{equation*}
\hat{\pi}_u^{\texttt{tl}} = 
     \begin{cases}
        \pi_u^{\texttt{tl}}       &~\text{if } c_1(\pi_u^{\texttt{tl}}) \leq c_1(\tilde{\pi}_u^{\texttt{tl}}) \\
        \tilde{\pi}_u^{\texttt{tl}}  &~\text{if } c_1(\tilde{\pi_u}^{\texttt{tl}}) < c_1(\pi_u^{\texttt{tl}}),         
     \end{cases}
\end{equation*}
and
\begin{equation*}
\hat{\pi}_u^{\texttt{br}} = 
     \begin{cases}
        \pi_u^{\texttt{br}}       &~\text{if } c_2(\pi_u^{\texttt{br}}) \leq c_2(\tilde{\pi}_u^{\texttt{br}}) \\
        \tilde{\pi}_u^{\texttt{br}}  &~\text{if } c_2(\tilde{\pi_u}^{\texttt{br}}) < c_2(\pi_u^{\texttt{br}}).         
     \end{cases}
\end{equation*}
For a visualization, see Fig.~\ref{fig:operations}.

We are finally ready to describe \is, our algorithm for bi-criteria approximate shortest-path computation (Problem~\ref{prob:2}).
We run a best-first search similar to Alg.~\ref{alg:astar} but nodes are path pairs.
We start with the trivial path pair $(\vs, \vs)$ and describe our algorithm by detailing the different functions highlighted in Alg.~\ref{alg:astar}.
For each function, we describe what needs to be performed and how this can be efficiently implemented when consistent heuristics are used (see Sec.~\ref{sec:background}).
Finally, the pseudocode of the algorithm is provided in Alg.~\ref{alg:is} with the efficient implementations provided in Alg.~\ref{alg:is_dominated}-\ref{alg:merge-IS}.

\begin{algorithm}[t!]
    \textbf{Input: ($G = (V,E), \vs, \vg, c_1, c_2, h_1, h_2, \eps_1, \eps_2$)}
        \begin{algorithmic}[1]
        \State {solutions\_pp$\gets \emptyset$}
        \Comment{path pairs }
        \State {${\rm OPEN}\gets$ new path pair $(\vs, \vs)$} 
        \vspace{1mm}
        \While{${\rm{OPEN}} \neq \emptyset$}
          \State $(\pi_u^{\texttt{tl}}, \pi_u^{\texttt{br}}) \gets$ \rm{OPEN{}.\func{extract\_min}}()
          \If {\func{is\_dominated\_\is}($\pi_u^{\texttt{tl}}, \pi_u^{\texttt{br}}$)}
            \State {\textbf{continue}}
          \EndIf
        \vspace{1mm}
          \If {$u=\vg$} \Comment{reached goal}
            \State {\func{merge\_to\_solutions\_\is}($\pi_u^{\texttt{tl}}, \pi_u^{\texttt{br}}$, solutions\_pp)}
            \State {\textbf{continue}}
          \EndIf
        \vspace{1mm}
          \For{$e=(u,v) \in$ {neighbors}($s(n), G$)}
            \State $\left(\pi_v^{\texttt{tl}}, \pi_v^{\texttt{br}}\right) \gets$ \func{extend\_\is}($(\pi_u^{\texttt{tl}}, \pi_u^{\texttt{br}}), e$)
          \If {\func{is\_dominated\_\is}($\pi_v^{\texttt{tl}}, \pi_v^{\texttt{br}}$)}
            \State {\textbf{continue}}
          \EndIf

        \vspace{1mm}
            \State {\func{insert\_\is}($(\pi_v^{\texttt{tl}}, \pi_v^{\texttt{br}}), \rm{OPEN}$)}
          \EndFor
        
        \EndWhile

       \State{solutions$\gets \emptyset$}
       \For{$(\pi_\vg^{\texttt{tl}}, \pi_\vg^{\texttt{br}}) \in$ solutions\_pp}
         \State solutions $\gets$  solutions $\cup\{  \pi_\vg^{\texttt{tl}} \}$
       \EndFor
       \State{\textbf{return} solutions}
\end{algorithmic}
    \caption{\is}
    \label{alg:is}
\end{algorithm}

\paragraph{Ordering nodes in OPEN:}
Recall that a node is a path pair 
$(\pi_u^{\texttt{tl}}, \pi_u^{\texttt{br}})$ and that each path $\pi$ has two $f$ values which correspond to the two cost functions and the two heuristic functions. 
Nodes are ordered lexicographically according to 
\begin{equation}
  \label{eq:lexi}
  \large(
  f_1(\pi_u^{\texttt{tl}}), 
  f_2(\pi_u^{\texttt{br}})
  \large).
\end{equation}
\paragraph{Domination checks:}
Recall that there are two types of domination checks that we wish to perform
(i)~checking if a node is dominated by a node that was already expanded
and
(ii)~checking if a node has the potential to reach the goal with a solution whose cost is not dominated by any existing solution. 

In our setting a path pair 
$\rm{PP}_u$
is dominated by another path pair
$\tilde{\rm{PP}}_u$
if the 
partial Pareto frontier represented by $\rm{PP}_u$ 
is contained in the 
partial Pareto frontier represented by $\tilde{\rm{PP}}_u$ (see Fig.~\ref{fig:ppf-dominates}).
We can efficiently test if 
$\rm{PP}_u = (\pi_u^{\texttt{tl}}, \pi_u^{\texttt{br}})$
 is dominated by any path to $u$ found so far, by checking if
\begin{equation}
\label{eq:d1}
g_2({\pi}_u^{\texttt{br}}) 
\geq 
g_2^{\rm min}(u).
\end{equation}
This only holds when using the assumption that our heuristic functions are admissible and consistent and using the way we order our OPEN list.

\begin{figure}[tb]
  \centering
  \includegraphics[height=5cm]{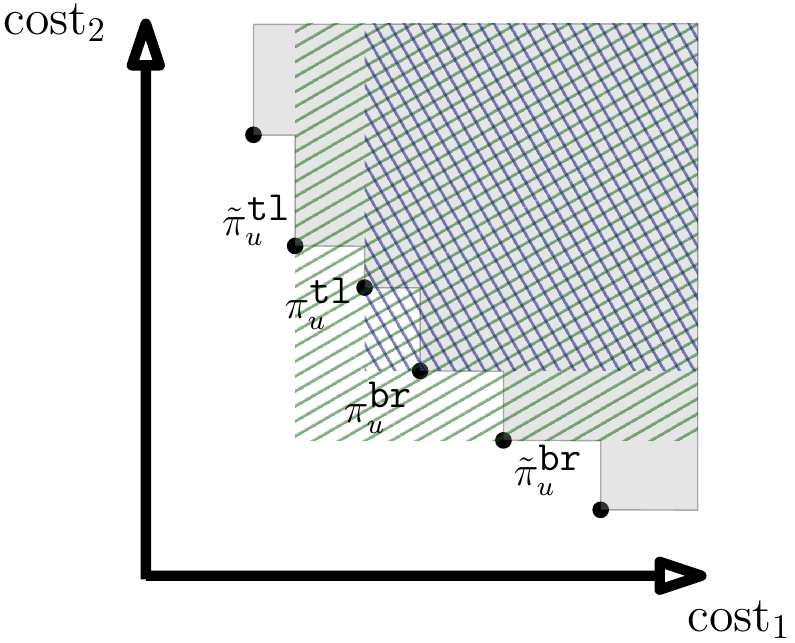}
  \caption{Testing dominance of partial Pareto frontiers using path pairs.
           The partial Pareto frontier $\Pi_u^{\pi_u^{\texttt{tl}}, \pi_u^{\texttt{br}}}$
           is contained in 
           the partial Pareto frontier $\Pi_u^{\tilde{\pi}_u^{\texttt{tl}}, \tilde{\pi}_u^{\texttt{br}}}$.
           Thus, the region represented by 
           ${\pi}_u^{\texttt{tl}}, {\pi}_u^{\texttt{br}}$
           is contained in the region represented by 
           $\tilde{\pi}_u^{\texttt{tl}}, \tilde{\pi}_u^{\texttt{br}}$.
           }
  \label{fig:ppf-dominates}
\end{figure}

We now continue to describe how we test if a path pair has the potential to reach the goal with a solution whose cost is not dominated by any existing solution.
Given a path pair
$\rm{PP}_u = (\pi_u^{\texttt{tl}}, \pi_u^{\texttt{br}})$ 
a lower bound on the partial Pareto frontier at $\vg$ that can be attained via $\rm{PP}_u$ is obtained by adding the heuristic values to the costs of the two paths in $\rm{PP}_u$.
Namely,
we consider two paths $\pi_{\vg}^{\texttt{tl}}, \pi_{\vg}^{\texttt{br}}$ such that
$
c_i(\pi_{\vg}^{\texttt{tl}}):= c_i(\pi_u^{\texttt{tl}}) + h_i(u)
$
and 
$
c_i(\pi_{\vg}^{\texttt{br}}):= c_i(\pi_u^{\texttt{br}}) + h_i(u)
$.
Note that these paths may not be attainable and are a lower bound on the partial Pareto frontier that can be obtained via~$\rm{PP}_u$.
Now, if the partial Pareto frontier 
$\ppf_{\vg}^{\pi_{\vg}^{\texttt{tl}}, \pi_{\vg}^{\texttt{br}}}$
is contained in the union of the currently-computed partial Pareto frontiers at $\vg$, then $\rm{PP}_u$ is dominated.
%
Similar to the previous dominance check, this can be efficiently implemented by testing if
\begin{equation}
\label{eq:d2}
(1 + \eps_2) \cdot (f_2(\pi_u^{\texttt{br}})) \geq g_2^{\rm min}(\vg).
\end{equation}

\paragraph{Inserting nodes in OPEN:}
Recall that we want to use the notion of path pairs to represent a partial Pareto frontier. 
Key to the efficiency of our algorithm is to have every partial Pareto frontier as large as possible under the constraint that they are all $(\eps_1, \eps_2)$-bounded.
Thus, when coming to insert a path pair 
$\rm{PP}_u$ 
into the OPEN list, 
we check if there exists a path pair
$\rm{\tilde{PP}}_u$ such that 
$\rm{PP}_u$ and $\rm{\tilde{PP}}_u$ can be merged and the resultant path pair is still $(\eps_1, \eps_2)$-bounded.

If this is the case, we remove $\rm{\tilde{PP}}_u$ and replace it with the merged path pair.

\paragraph{Merging solutions:}
Since we want to minimize the number of path pairs representing $\Pi_\vg(\eps_1, \eps_2)$ we suggest an optimization that operates similarly to node insertions.
When a new path pair $\rm{PP}_\vg$  representing a partial Pareto frontier at $\vg$ is obtained, we test if there exists a path pair in the solution set
$\rm{\tilde{PP}}_\vg$ such that 
$\rm{PP}_u$ and $\rm{\tilde{PP}}_\vg$ can be merged and the resultant path pair is still $(\eps_1, \eps_2)$-bounded.

If this is the case, we remove $\rm{\tilde{PP}}_\vg$ and replace it with the merged path pair.

\paragraph{Returning solutions:}
Recall that our algorithm stores solutions as path pairs and not individual paths.
To return an approximate Pareto frontier, we simply return one path in each path pair.
Here, we arbitrarily choose to return $\pi_\vg^{\texttt{tl}}$ for each path pair $(\pi_\vg^{\texttt{tl}}, \pi_\vg^{\texttt{br}})$.


\begin{algorithm}[t!]
    \textbf{Input:} ($\rm{PP}_u = (\pi_u^{\texttt{tl}}, \pi_u^{\texttt{br}})$)
    \begin{algorithmic}[1]
      \If{$(1 + \eps_2) \cdot f_2(\pi_u^{\texttt{br}}) \geq g_2^{\rm min}
            (\vg)$}
        \State{\textbf{return} \texttt{true}}
        \Comment{dominated by solution}
      \EndIf
      \If{$g_2({\pi}_u^{\texttt{br}})\geq g_2^{\rm min}(u)$}
        \State{\textbf{return} \texttt{true}}
        \Comment{dominated by existing path pair}
      \EndIf
      \State{\textbf{return} \texttt{false}}

\end{algorithmic}
    \caption{is\_dominated\_\is}
    \label{alg:is_dominated}
\end{algorithm}

\begin{algorithm}[t!]
    \textbf{Input:} ($\rm{PP}_u = (\pi_u^{\texttt{tl}}, \pi_u^{\texttt{br}}) , e = (u,v)$)
    \begin{algorithmic}[1]
      \State{$\pi_v^{\texttt{tl}} \gets $\func{extend}($\pi_u^{\texttt{tl}}$)}
      \State{$\pi_v^{\texttt{br}} \gets $\func{extend}($\pi_u^{\texttt{br}}$)}
      \State{\textbf{return} $(\pi_v^{\texttt{tl}}, \pi_v^{\texttt{br}})$}
\end{algorithmic}
    \caption{extend\_\is}
    \label{alg:extendIS}
\end{algorithm}

\begin{algorithm}[t!]
    \textbf{Input:} ($\rm{PP}_v$, OPEN)
    \begin{algorithmic}[1]
      \For{\textbf{each} path pair $\rm{\tilde{PP}}_v \in$ OPEN}
        \State $\rm{PP}_v^{\rm merged} \gets $ merge($\rm{\tilde{PP}}_v, \rm{{PP}}_v $)
        \If {$\rm{PP}_v^{\rm merged}$.is\_bounded($\eps_1, \eps_2$)}
          \State OPEN.remove($\rm{\tilde{PP}}_v$)
          \Comment{remove existing path pair}
          \State OPEN.\func{insert}($\rm{PP}_v^{\rm merged}$)
          \State \textbf{return}
        \EndIf 
        \EndFor
       \State OPEN.\func{insert}($\rm{PP}_v$)
       \State{\textbf{return}}
\end{algorithmic}
    \caption{insert\_\is}
    \label{alg:insertIS}
\end{algorithm}

\begin{algorithm}[t!]
    \textbf{Input:} ($\rm{PP}_{\vg}$, solutions\_pp)
    \begin{algorithmic}[1]
      \For{\textbf{each} path pair $\rm{\tilde{PP}}_{\vg} \in$ solutions\_pp}
        \State $\rm{PP}_{\vg}^{\rm merged} \gets $ merge($\rm{\tilde{PP}}_{\vg}, \rm{{PP}}_v $)
        \If {$\rm{PP}_{\vg}^{\rm merged}$.is\_bounded($\eps_1, \eps_2$)}
          \State solutions\_pp.remove($\rm{\tilde{PP}}_{\vg}$)
          \State solutions\_pp.insert($\rm{PP}_{\vg}^{\rm merged}$)
          \State \textbf{return}
        \EndIf 
        \EndFor
       \State solutions\_pp.insert($\rm{PP}_{\vg}$)
       \State{\textbf{return}}
\end{algorithmic}
    \caption{merge\_to\_solutions\_\is}
    \label{alg:merge-IS}
\end{algorithm}
\subsection{Analysis}
To show that \is indeed computes an approximate Pareto frontier using the  domination checks suggested in Eq.~\ref{eq:d1} and~\ref{eq:d2}, it will be useful to introduce the notion of a path pair's \myemph{apex} which represents a (possibly non-existent) path that dominates all paths represented by the path pair.
\begin{defn}
Let $\rm{PP}_u = (\pi_u^{\texttt{tl}}, \pi_u^{\texttt{br}})$ be a path pair.
Its apex is a two-dimensional point 
$\calA_u = 
(
c_1(\pi_u^{\texttt{tl}}),
c_2(\pi_u^{\texttt{br}}
)$.
The $f$-value of the apex is 
$(
c_1(\pi_u^{\texttt{tl}}) + h_1(u),
c_2(\pi_u^{\texttt{br}} + h_2(u)
)$
\end{defn}

If we (conceptually) replace each path pair used by \is with its corresponding apex, we obtain an algorithm that is very similar to \boastareps.
Specifically, both algorithms 
(i)~order nodes / apexes in the OPEN list lexicographically,
(ii)~update $g_2^{\rm {min}}(u)$ for each vertex $u$ when a node / apex is popped if $g_2 < g_2^{\rm {min}}(u)$ where~$g_2$ is the $c_2$-cost of the node / apex
and
(iii)~prune a node / apex using the same condition (notice that Eq.~\ref{eq:d0} is identical to Eq~\ref{eq:d1} and~\ref{eq:d2}). 

The main difference is that \is also merges path pairs. 
This turns the search tree into a directed acyclic graph (DAG).

Now we can easily adapt several Lemmas described in~\cite{UYBZSK20}\footnote{When adapting the Lemmas used in~\cite{UYBZSK20}, we mention the corresponding Lemma number. Furthermore, to ease a reader interested in comparing the Lemmas and the corresponding proofs, when possible, we use the same notation and even the same words. This was done after obtaining permission from the authors in~\cite{UYBZSK20}.}:
\begin{lem}[Corresponding to Lemma~1 in~\cite{UYBZSK20}]
\label{lem:1}
After extending an apex, the new apex has $f_1$- and $f_2$-values that are no smaller than the $f_1$- and $f_2$-values, respectively, of the generating apex.
\end{lem}
The proof is identical to Lemma~1 in~\cite{UYBZSK20}. 
However, it is important to note that an apex can also be created by merging two existing apexes. Here, it is not clear which is the ``parent'' apex and the Lemma does not necessarily hold.
However, we can state the following straightforward obsrevation:

\begin{observation}
\label{obs:1}
Let $\calA_u^{\rm merged}$ by the apex resulting from a merge operation between $\calA_u^1$ and $\calA_u^2$.
Then, the   $f_1$- and $f_2$-values of $\calA_u^{\rm merged}$ cannot be smaller than the $f_1$- and $f_2$-values of both of $\calA_u^1$ and $\calA_u^2$.
Specifically,
$$
  f_1(\calA_u^{\rm merged}) = \min (f_1(\calA_u^1),f_1( \calA_u^2)),
$$
and
$$
  f_2(\calA_u^{\rm merged}) = \min (f_2(\calA_u^1), f_2(\calA_u^2)).
$$
\end{observation}

\begin{lem}[Corresponding to Lemma~2 in~\cite{UYBZSK20}]
\label{lem:2}
The sequences of extracted and expanded apexes have monotonically non-decreasing $f_1$-values.
\end{lem}
\begin{proof}
 An apex extracted by \is from the OPEN list  has the smallest $f_1$-value among of all apexes in the Open list.
  Since generated apexes that are
added to the Open list have $f_1$-values that are no smaller than those of their expanded parent apexes (Lemma~\ref{lem:1}) and an apex resulting from a merge operation  cannot have an $f_1$-value smaller than the $f_1$- value of both of the merged apexes (Obs~\ref{obs:1}), the sequence of extracted apexes has monotonically non-decreasing $f_1$-values.  
\end{proof}

\begin{lem}[Corresponding to Lemma~3 in~\cite{UYBZSK20}]
\label{lem:3}
The sequences of extracted apexes with the same state has strictly monotonically decreasing $f_2$-values.
\end{lem}
The proof is identical to Lemma~3 in~\cite{UYBZSK20}. 

\begin{lem}[Corresponding to Lemma~6 in~\cite{UYBZSK20}]
\label{lem:6}
If 
apex $\calA_u^1$
is weakly dominated by 
apex $\calA_u^2$, 
then each apex at $\vg$ in
the subtree of the search tree rooted at $\calA_u^1$ 
(when no merge operations are performed in the subtree) is weakly
dominated by an apex at $\vg$ in the DAG rooted
at $\calA_u^2$
(even when merge operations are performed).
\end{lem}
\begin{proof}
Since apex $\calA_u^1$
is weakly dominated by 
apex $\calA_u^2$,
$g_1(\calA_u^1) \leq g_1(\calA_u^2)$ 
(here, given apex $\calA = (p_1,p_2)$ we define  $g_i(\calA): = p_i$ corresponding to the $g$-value of nodes in standard \astar-like algorithms). 
Assume that $\calA_{\vg}^3$ is an apex at the goal in the subtree of the search tree rooted at~$\calA_u^1$ (when no merge operations are performed). 
Let the sequence of vertices of the apexes along a branch of the search tree from the root apex to $\calA_u^1$ be 
$u_1, \ldots u_i$
(with $u_1 = \vs$ and $u_i = u$).
Similarly, let the sequence of vertices of the apexes  along a branch of the DAG from the root apex to $\calA_u^2$ be 
$u_1', \ldots u_j'$
(with $u_1' = \vs$ and $u_j' = u$).
Finally, let the sequence of vertices of the apexes  along a branch of the search tree from $\calA_u^1$ to $\calA_{\vg}^3$ be
$\pi = u_i, \ldots u_k$
(with $u_k = \vg$).

Then, there is an apex $\calA_{\vg}^4$ at the goal in the DAG rooted at apex $\calA_u^2$ such that the sequence of vertices of the apexes along a branch of the DAG from
the root apex to~$\calA_{\vg}^4$ is 
$u_1', \ldots u_j', u_{i+1}, \ldots ,u_k$.
Since
$g_1(\calA_u^1) \leq g_1(\calA_u^2)$, it follows that 
\begin{align*}
g_1(\calA_{\vg}^4) 
  ~= &~ 
  g_1(\calA_u^2)+c_1(\pi) \\
  ~\leq &~
  g_1(\calA_u^1) + c_1(\pi) \\
  ~= &~ g_1(\calA_{\vg}^3).
\end{align*}
Thus,  
$g_1(\calA_{\vg}^4) \leq g_1(\calA_{\vg}^3)$.
Following similar lines yields that
$g_2(\calA_{\vg}^4) \leq g_2(\calA_{\vg}^3)$.
Thus, $\calA_{\vg}^4$  weakly dominates $\calA_{\vg}^3$.
\end{proof}

The following Lemma concludes the building blocks we will need to prove that \is computes an approximate Pareto frontier. However, it can be easily adapted to show that \boastar also computes an approximate Pareto frontier.

\begin{lem}[Corresponding to Lemma~7 in~\cite{UYBZSK20}]
\label{lem:7}
When \is prunes an apex $\calA_u^1$ at state $u$  and this prevents it in the future from adding an apex $\calA_{\vg}^2$ (at the goal state) to the solution set, then it can still add in the future an apex (with the goal state) that approximately dominates apex $\calA_{\vg}^2$.
\end{lem}
\begin{proof}
We prove the statement by induction on the number of pruned apexes so far, including apex $\calA_u^1$. If the
number of pruned apexes is zero, then the lemma trivially
holds. Now assume that the number of pruned apexes is $n+1$ and the lemma holds for $n \geq 0$. We distinguish three cases:
\begin{itemize}
\item[\textbf{C1}]
\label{case:1}
\is prunes apex $\calA_u^1$ on line~5 because of the first
pruning condition (Eq.~\ref{eq:d1}): 
Then, \is has expanded an apex~$\calA_u^4$ at state $u$ previously such that 
$g_2^{\rm {min}}(u) = g_2(\calA_u^4)$
since otherwise 
$g_2^{\rm {min}}(u) = \infty$
and the pruning condition could not hold. 
Combining both (in)equalities yields 
$g_2(\calA_u^1) \geq  g_2(\calA_u^4)$.
Since 
$f_1(\calA_u^1) \geq f_1(\calA_u^4)$
(Lemma~\ref{lem:2}),
\begin{align*}
  g_1(\calA_{u}^1) + h_1(u) 
  ~= &~ 
  f_1(\calA_u^1) \\
  ~\geq &
  f_1(\calA_u^4) \\
  ~= &~ 
  g_1(\calA_{u}^4) + h_1(u).
\end{align*}
Thus 
$g_1(\calA_{u}^1) \geq g_1(\calA_{u}^4)$. 
Combining both inequalities yields that apex $\calA_{u}^1$ is weakly dominated by apex $\calA_{u}^4$ and thus each apex at $\vg$ in the subtree rooted at $\calA_{u}^1$, including~$\calA_{\vg}^2$, is weakly dominated (and hence approximately dominated) by an apex $\calA_{\vg}^5$ in the subtree rooted at apex $\calA_{u}^4$ (Lemma~\ref{lem:6}). 
In case \is has pruned an apex that prevents it in the future from adding apex $\calA_{u}^5$ to the solution set, then it can still add in the future an apex 
(at the goal state) that approximately dominates~$\calA_{u}^5$ and thus also apex $\calA_{u}^2$ (induction assumption).

\item[\textbf{C2}] 
\label{case:2}
\is prunes apex $\calA_u^1$ on line~5 because of the second
pruning condition (Eq.~\ref{eq:d2}):
Then, \is has expanded an apex $\calA_{\vg}^4$  with the goal state previously such that 
$g_2^{\rm {min}}(\vg) = g_2(\calA_{\vg}^4)$
since otherwise 
$g_2^{\rm {min}}(\vg) = \infty$
and the pruning condition could not hold. 
Combining both (in)equalities yields that 
$(1 + \eps_2) \cdot f_2(\calA_u^1) \geq g_2(\calA_{\vg}^4)$.
Since $\calA_u^1$ is an ancestor of $\calA_{\vg}^2$ in the search tree, 
$f_2(\calA_{\vg}^2) \geq f_2(\calA_u^1)$ (Lemma~\ref{lem:1}). 
Combining both inequalities yields 
$g_2(\calA_{\vg}^2) = f_2(\calA_{\vg}^2) \geq g_2(\calA_{\vg}^4) / (1 + \eps_2)$.
Since $\calA_{u}^1$ is an ancestor of $\calA_{\vg}^2$ in the search tree, 
$g_1(\calA_{\vg}^1) = f_1(\calA_{\vg}^2) \geq f_1(\calA_u^1)$ (Lemma~\ref{lem:1}).
Since $f_1(\calA_u^1) \geq f_1(\calA_{\vg}^4)$ (Lemma~\ref{lem:2}), it follows that 
$g_1(\calA_{\vg}^2) \geq f_1(\calA_u^1) \geq f_1(\calA_{\vg}^4) = g_1(\calA_{\vg}^4)$. 
Combining 
$g_1(\calA_{\vg}^2) \geq g_1(\calA_{\vg}^4)$
and 
$g_2(\calA_{\vg}^2) \geq g_2(\calA_{\vg}^4) / (1 + \eps_2)$.
yields that 
$\calA_{\vg}^2$ 
is approximately dominated by 
$\calA_{\vg}^4$. 
In case \is has pruned an apex that prevents it in the future from adding $\calA_{\vg}^4$ to the solution set, then it can still add in the future an apex that approximately dominates $\calA_{\vg}^4$ and thus also $\calA_{\vg}^2$(induction assumption).
\item[\textbf{C3}] 
\is prunes apex $\calA_u^1$ on line~12 because of either the first or the second pruning condition (Eq.~\ref{eq:d1} or~\ref{eq:d2}):
The proofs of Case~\textbf{C1} or Case~\textbf{C2}, respectively, apply unchanged except that 
$f_1(\calA_u^1) \geq f_1(\calA_{\vg}^4)$
 now holds for a different reason. 
Let $\calA_v^3$ be the apex that \is expands when it executes Line 12. 
Combining 
$f_1(\calA_u^1) \geq f_1(\calA_v^3)$ (Lemma~\ref{lem:1}) and $f_1(\calA_v^3) \geq f_1(\calA_{\vg}^4)$ (Lemma~\ref{lem:2}) yields
$f_1(\calA_u^1) \geq f1(\calA_{\vg}^4)$.
\end{itemize}
\end{proof}

Lemma~\ref{lem:7} states that all solutions in the partial Pareto frontier captured by a path pair that was pruned by \is are approximately dominated by an apex of a path pair that was used as a solution.
Combining this with the fact that all path pairs are $(\eps_1, \eps_2)$-bounded by constructed we obtain the following Corollary.

\begin{cor}
  Given a Bi-criteria approximate shortest path (Problem~\ref{prob:2}), \is returns an approximate Pareto frontier.
\end{cor}

\section{Evaluation}
\paragraph{Experimental setup.}
To evaluate our approach we compare it to \boastareps as \boastar
was recently shown to dramatically outperform other state-of-the-art algorithms for bi-criteria shortest path~\cite{UYBZSK20}.
%
All experiments were run on an
1.8GHz Intel(R) Core(TM) i7-8565U CPU  
Windows 10 machine with 16GB of RAM.
All algorithm implementations were in \Cpp.\footnote{Our code is publicly available at \url{https://github.com/CRL-Technion/path-pair-graph-search}.}
We use road maps from the 9'th DIMACS Implementation Challenge: Shortest Path\footnote{\url{http://users.diag.uniroma1.it/challenge9/download.shtml}.}.
The cost components represent travel distances ($c_1$)
and times ($c_2$). The heuristic values are the exact travel distances and times to the goal state, computed with Dijkstra’s
algorithm. Since all algorithms use the same heuristic values, heuristic-computation times are omitted. 

\begin{table}[t]
 \input{tbl_NY}
 \input{tbl_BAY}
 \input{tbl_COL}
 \input{tbl_FLA}
\caption{Average number of solutions ($n_{\rm{sol}}$) and runtime (in ms) comparing \boastareps and \is on 50 random queries sampled for four different roadmaps for different approximation factors.}
\label{tbl:res}
\end{table}

\begin{figure*}[t]%
\centering
  \subfloat[]{
    \label{fig:NE1}
    \includegraphics[height=3.7cm]{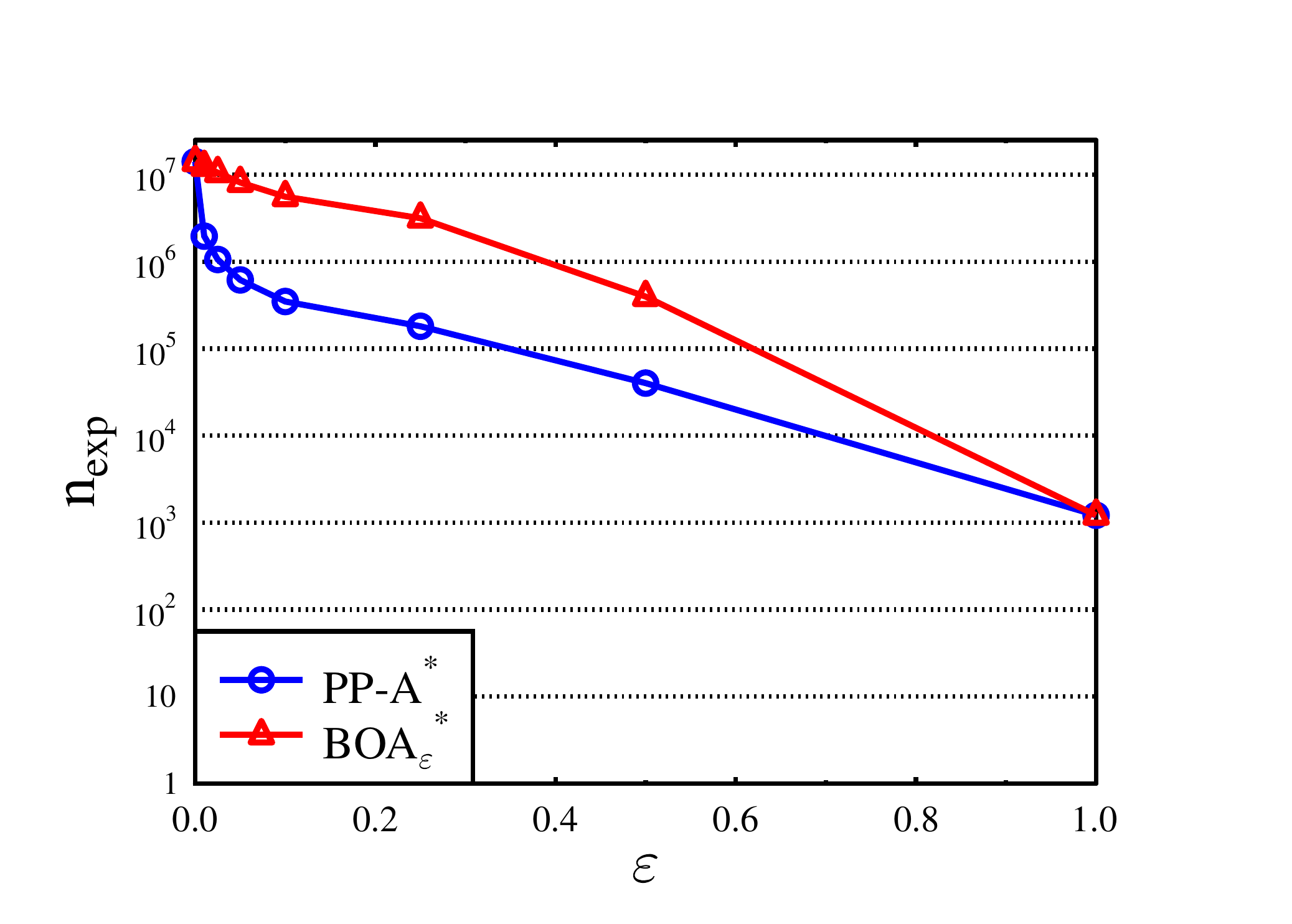}
  }
  \hspace*{-10mm}
  \subfloat[]{
    \label{fig:NE2}
    \includegraphics[height=3.7cm]{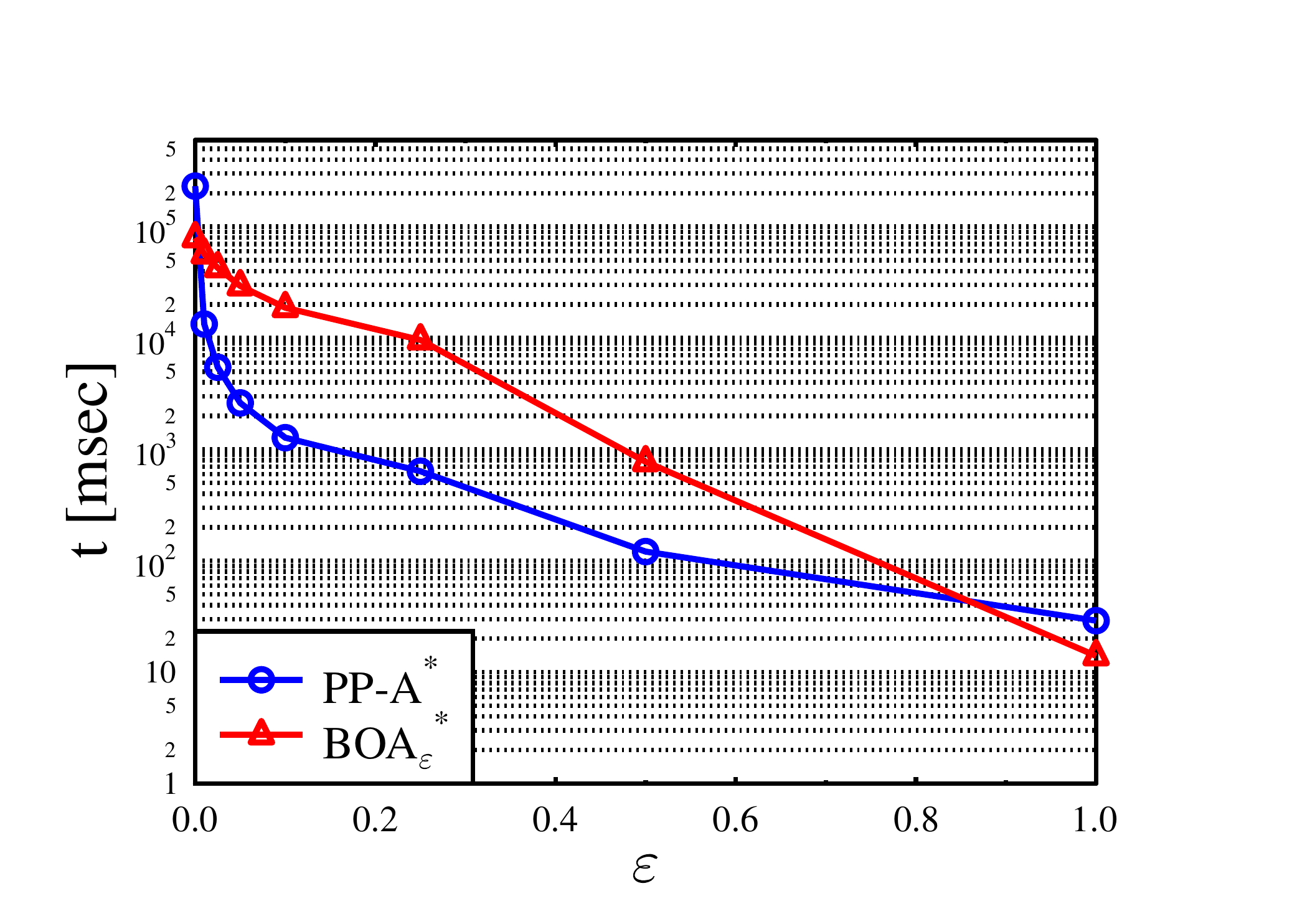}
  }
  \hspace*{-10mm}
  \subfloat[]{
    \label{fig:NE4}
    \includegraphics[height=3.7cm]{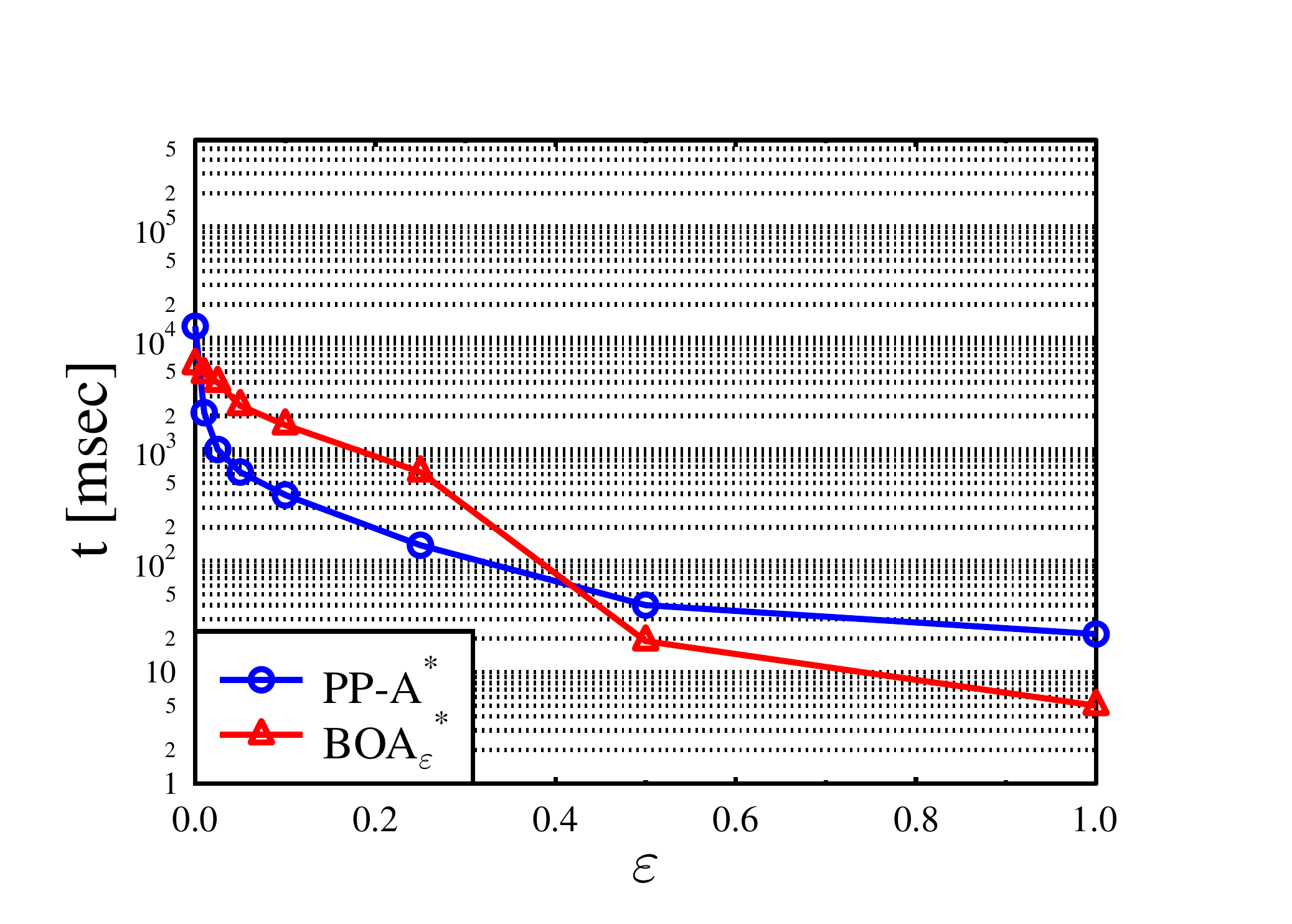}
  }  \hspace*{-10mm}
  \subfloat[]{
    \label{fig:NE3}
    \includegraphics[height=3.7cm]{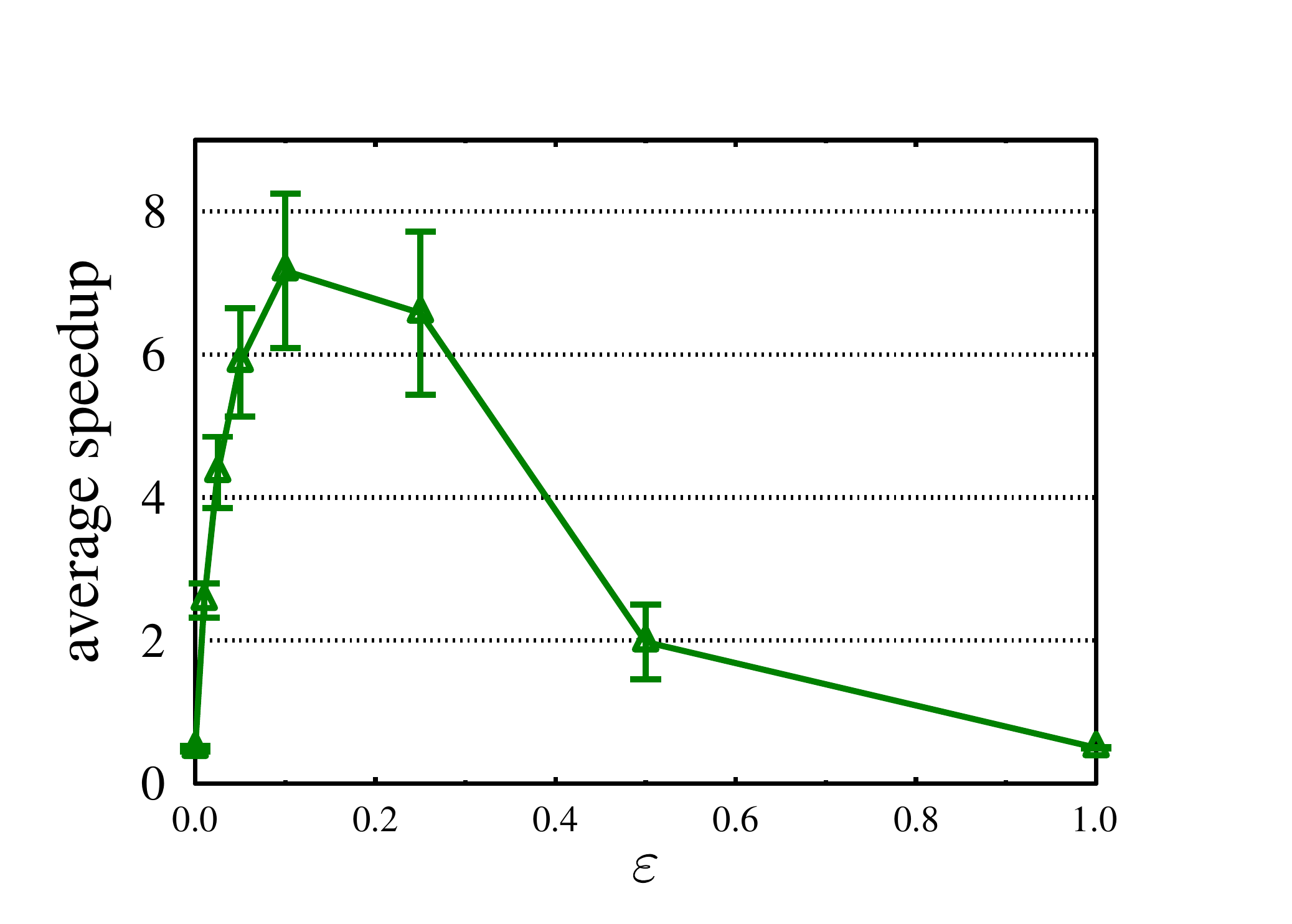}
  }%
  \caption{North East (NE) plots.
  \protect \subref{fig:NE1}
  The average number of expanded nodes ($n_{\rm{exp}}$). 
  \protect \subref{fig:NE2} and 
  \protect \subref{fig:NE4}
  the time (arithmetic-mean and geometric mean, respectively) for both algorithms as a function of the approximation factor. 
  Notice the logarithmic scale in the $y$-axis for the first three plots.
  \protect \subref{fig:NE3}
  The average speedup of \is when compared to \boastareps as a function of the approximation factor. 
  Error bars denote one standard error (error bars in~\protect \subref{fig:NE1} through \protect \subref{fig:NE4} are not visible due to the logarithmic scale).
  }
\label{fig:NE}
\end{figure*}

\paragraph{General comparison.}Similar to the experiments of Hernandez et al~\cite{UYBZSK20} we start by comparing the algorithms for four different roadmaps containing between roughly $250K$ and $1M$ vertices. 
Table~\ref{tbl:res} summarizes the number of solutions in the approximate Pareto frontier and average, minimum and maximum running times of the two algorithms using the following values\footnote{While \is allows a user to specify two approximation factors corresponding to the two cost functions, this is not the case for \boastar. Thus, in all experiments we use a single approximation factor $\varepsilon$ and set $\varepsilon_1 = \varepsilon_2 = \varepsilon$.} 
$\varepsilon \in \{ 0, 0.01, 0.025, 0.05, 0.1\}$.
Here, approximation values of zero and $0.01$ correspond to computing the entire Pareto frontier and approximating it using a value of~$1\%$, respectively.

When computing the entire Pareto frontier \boastar is roughly three times faster than \is   on average. 
This is to be expected as \is stores for each element in the priority queue two paths and requires more computationally-demanding operations.
As the approximation factor is increased, the average running time of \is drops faster, when compared to \boastareps and we observe a significant average speedup. 
%
Interestingly, when looking at the minimal running time, \boastareps significantly outperforms \is. This is because in such settings the approximate Pareto frontier contains one solution, which \boastareps is able to compute very fast. Other nodes are approximately dominated by this solution and the algorithm can terminate very quickly.
\is, on the other hand, still performs merge operations which incur a computational overhead.
When looking at the maximal running time, we can see an opposite trend where \is outperforms \boastareps by a large factor.
%

\paragraph{Pinpointing the performance differences between \is and \boastareps.}
The first set of results suggest that as the problem becomes harder, the speedup that \is may offer becomes more pronounced.
We empirically quantify this claim by moving to a larger map called the North East (NE) map which contains 1,524,453 states and 3,897,636 edges where we obtain even larger speedups (see Table~\ref{tbl:res-NE}).

\begin{table}[t]
 \input{tbl_NE}
\caption{Runtime (in seconds) comparing \boastar and \is on 50 random queries sampled for the NE map.}
\label{tbl:res-NE}
\end{table}

We plot the number of nodes expanded (which typically is proportional to the running time of \astar-like algorithms) of each algorithm as a function of the approximation factor (see, Fig.~\ref{fig:NE1}.
Here we used $\varepsilon \in \{ 0, 0.01, 0.025, 0.05, 0.1, 0.25, 0.5, 1\}$.
Additionally, we plot both the arithmetic mean (Fig.~\ref{fig:NE2})
as well as the 
geometric mean (Fig.~\ref{fig:NE4})
of each algorithm as a function of the approximation factor.\footnote{We used both arithmetic and geometric mean as the arithmetic mean can be misleading skewing the mean towards the results on larger instances. Together, both means better capture the results. }

We observe that the number of nodes expanded monotonically decreases when the approximation  factor is increased for both algorithms. This is because additional nodes may be pruned which in turn, prunes all nodes in their subtree.
It is important to discuss \emph{how} these nodes are pruned: Recall that \boastareps prunes nodes according to Eq.~\ref{eq:d0}.
Thus, increasing the approximation factor only allows to prune more nodes according to the already-computed solutions and not according to the paths computed to intermediate nodes.
In contrast, \is prunes nodes according to Eq.~\ref{eq:d1} and~\ref{eq:d2}. Thus, in addition to more path pairs being merged,  increasing the approximation allows to prune more path pairs according to the already-computed solutions as well as the path pairs computed to intermediate vertices.
Thus, for relatively-small approximation factors that are greater than zero (in our setting, $0 < \varepsilon < 0.5$, we see that \boastar expands a significantly higher number of nodes than \is which explains the speedups we observed.
However, for large approximation factors, there is typically only one solution in the approximate Pareto frontier. This solution, which is found quickly by  \boastareps,
allows to prune almost all other paths which results in \boastareps running much faster than \is.
This trend is visualized in Fig.~\ref{fig:NE3}.

\section{Future Research}


\subsection{Bidirectional search}
We presented \is as a unidirectional search algorithm, however a common approach to speed up search algorithms is to perform  two simultaneous searches: a forward search from $\vs$ to $\vg$ and a backward search from $\vg$ to~$\vs$~\cite{pohl1971bi}.
Thus, an immediate task for future research is to suggest a bidirectional extension of \is. Here we can build upon recent progress in bi-directional search algorithms for bi-criteria shortest-path problems~\cite{sedeno2019biobjective}.
\subsection{Beyond two optimization criteria}
We presented \is as a search algorithm for two optimization criteria, however the same concepts can be used for multi-criteria optimization problems. Unfortunately, it is not clear how to perform  operations such as dominance checks efficiently since the methods presented for \boastar do not extend to such settings.

\section*{Acknowledgements}
We wish to thank 
Carlos Hernandez, William Yeoh, Jorge  Baier and             Sven Koenig for insightful discussions regarding BOA*
and Ariel Felner for comments on early drafts of this paper.
In addition, we thank the anonymous reviewers of the ICAPS 2020 Workshop on Heuristics and Search for Domain-independent Planning (HSDIP 2020) for insightful comments on an early version of this paper.

Finally, this research was partially supported by grants No. 102583, 2028142 from the Isaeli Ministry of Science \& Technology (MOST), and by grant No. 1018193 from the United States-Israel Binational Science Foundation (BSF).


\end{document}

%% file: macros.tex
\newcommand{\calA}{\ensuremath{\mathcal{A}}\xspace}

\newcommand{\R}{\mathbb{R}}

\newcommand{\eps}{\ensuremath{\varepsilon}\xspace}

\newtheorem{thm}{Theorem}
\newtheorem{lem}{Lemma}
\newtheorem{prop}{Property}
\newtheorem{observation}[thm]{Observation}
\newtheorem{cor}{Corollary}
\newtheorem{defn}[thm]{Definition}
\newtheorem{prob}{Problem}

\newcommand{\ignore}[1]{}

\newcommand{\myemph}[1]{{\color{blue}\emph{#1}}}
\newcommand{\func}[1]{{\color{teal}\textbf{#1}}}
\newcommand\algname[1]{\textsf{#1}\xspace}

\newcommand{\Cpp}{C\raise.08ex\hbox{\tt ++}\xspace}

\newcommand{\NP}{{\small \ensuremath{\mathsf{NP}}\xspace}}

\newcommand{\ppf}{{\small \ensuremath{\rm{PPF}}\xspace}}

\newcommand{\vs}{{\ensuremath{v_{\rm{start}}}\xspace}}
\newcommand{\vg}{{\ensuremath{v_{\rm{goal}}}\xspace}}

\newcommand{\astar}{\algname{A$^*$}}
\newcommand{\astareps}{\algname{A$^*_\eps$}}
\newcommand{\boastar}{\algname{BOA$^*$}}
\newcommand{\boastareps}{\algname{BOA$^*_\eps$}}

\newcommand{\is}{\algname{PP-A$^*$}}

%% file: tbl_NY.tex
\resizebox{\columnwidth}{!}{%
\begin{tabular}{cllllllll}
\rowcolor[HTML]{656565} 
\multicolumn{9}{c}{\cellcolor[HTML]{656565}\textbf{New York City (NY)}}                                                                                                                                                                                                                                                                                                                                                                                                                                                                                  \\
\rowcolor[HTML]{9B9B9B} 
\multicolumn{9}{c}{\cellcolor[HTML]{9B9B9B}\textbf{264,346 states, 730,100 edges}}                                                                                                                                                                                                                                                                                                                                                                                                                                                                       \\
\rowcolor[HTML]{C0C0C0} 
\cellcolor[HTML]{C0C0C0}                                         & \multicolumn{2}{c}{\cellcolor[HTML]{C0C0C0}\textbf{avg $n_{sol}$}}                                                  & \multicolumn{2}{c}{\cellcolor[HTML]{C0C0C0}\textbf{avg t}}                                                          & \multicolumn{2}{c}{\cellcolor[HTML]{C0C0C0}\textbf{min t}}                                                          & \multicolumn{2}{c}{\cellcolor[HTML]{C0C0C0}\textbf{max t}}                                                          \\
\rowcolor[HTML]{C0C0C0} 
\multirow{-2}{*}{\cellcolor[HTML]{C0C0C0}\textbf{$\varepsilon$}} & \multicolumn{1}{c}{\cellcolor[HTML]{C0C0C0}\textbf{\is}} & \multicolumn{1}{c}{\cellcolor[HTML]{C0C0C0}\textbf{\boastar}} & \multicolumn{1}{c}{\cellcolor[HTML]{C0C0C0}\textbf{\is}} & \multicolumn{1}{c}{\cellcolor[HTML]{C0C0C0}\textbf{\boastar}} & \multicolumn{1}{c}{\cellcolor[HTML]{C0C0C0}\textbf{\is}} & \multicolumn{1}{c}{\cellcolor[HTML]{C0C0C0}\textbf{\boastar}} & \multicolumn{1}{c}{\cellcolor[HTML]{C0C0C0}\textbf{\is}} & \multicolumn{1}{c}{\cellcolor[HTML]{C0C0C0}\textbf{\boastar}} \\
\textbf{0}                                                       & {\color[HTML]{000000} 158}                              & {\color[HTML]{000000} 158}                                & {\color[HTML]{000000} 1,047}                            & {\color[HTML]{000000} 405}                                & {\color[HTML]{000000} 2}                                & {\color[HTML]{000000} 0}                                  & {\color[HTML]{000000} 13,563}                           & {\color[HTML]{000000} 5,038}                              \\
\textbf{0.01}                                                    & {\color[HTML]{000000} 19}                               & {\color[HTML]{000000} 20}                                 & {\color[HTML]{000000} 291}                              & {\color[HTML]{000000} 353}                                & {\color[HTML]{000000} 3}                                & {\color[HTML]{000000} 0}                                  & {\color[HTML]{000000} 3,662}                            & {\color[HTML]{000000} 4,577}                              \\
\textbf{0.025}                                                   & {\color[HTML]{000000} 10}                               & {\color[HTML]{000000} 10}                                 & {\color[HTML]{000000} 168}                              & {\color[HTML]{000000} 295}                                & {\color[HTML]{000000} 2}                                & {\color[HTML]{000000} 0}                                  & {\color[HTML]{000000} 2,207}                            & {\color[HTML]{000000} 4,101}                              \\
\textbf{0.05}                                                    & {\color[HTML]{000000} 6}                                & {\color[HTML]{000000} 6}                                  & {\color[HTML]{000000} 111}                              & {\color[HTML]{000000} 240}                                & {\color[HTML]{000000} 3}                                & {\color[HTML]{000000} 0}                                  & {\color[HTML]{000000} 1,523}                            & {\color[HTML]{000000} 3,538}                              \\
\textbf{0.1}                                                     & {\color[HTML]{000000} 4}                                & {\color[HTML]{000000} 4}                                  & {\color[HTML]{000000} 69}                               & {\color[HTML]{000000} 174}                                & {\color[HTML]{000000} 2}                                & {\color[HTML]{000000} 0}                                  & {\color[HTML]{000000} 932}                              & {\color[HTML]{000000} 2,694}                             
\end{tabular}
}

%% file: tbl_BAY.tex
\resizebox{\columnwidth}{!}{%
\begin{tabular}{cllllllll}
\rowcolor[HTML]{656565} 
\multicolumn{9}{c}{\cellcolor[HTML]{656565}\textbf{San Francisco Bay (BAY)}}                                                                                                                                                                                                                                                                                                                                                                                                                                                                                  \\
\rowcolor[HTML]{9B9B9B} 
\multicolumn{9}{c}{\cellcolor[HTML]{9B9B9B}\textbf{321,270 states, 794,830 edges}}                                                                                                                                                                                                                                                                                                                                                                                                                                                                       \\
\rowcolor[HTML]{C0C0C0} 
\cellcolor[HTML]{C0C0C0}                                         & \multicolumn{2}{c}{\cellcolor[HTML]{C0C0C0}\textbf{avg $n_{sol}$}}                                                  & \multicolumn{2}{c}{\cellcolor[HTML]{C0C0C0}\textbf{avg t}}                                                          & \multicolumn{2}{c}{\cellcolor[HTML]{C0C0C0}\textbf{min t}}                                                          & \multicolumn{2}{c}{\cellcolor[HTML]{C0C0C0}\textbf{max t}}                                                          \\
\rowcolor[HTML]{C0C0C0} 
\multirow{-2}{*}{\cellcolor[HTML]{C0C0C0}\textbf{$\varepsilon$}} & \multicolumn{1}{c}{\cellcolor[HTML]{C0C0C0}\textbf{\is}} & \multicolumn{1}{c}{\cellcolor[HTML]{C0C0C0}\textbf{\boastar}} & \multicolumn{1}{c}{\cellcolor[HTML]{C0C0C0}\textbf{\is}} & \multicolumn{1}{c}{\cellcolor[HTML]{C0C0C0}\textbf{\boastar}} & \multicolumn{1}{c}{\cellcolor[HTML]{C0C0C0}\textbf{\is}} & \multicolumn{1}{c}{\cellcolor[HTML]{C0C0C0}\textbf{\boastar}} & \multicolumn{1}{c}{\cellcolor[HTML]{C0C0C0}\textbf{\is}} & \multicolumn{1}{c}{\cellcolor[HTML]{C0C0C0}\textbf{\boastar}} \\
\textbf{0}                                                       & {\color[HTML]{000000} 117}                              & {\color[HTML]{000000} 117}                                & {\color[HTML]{000000} 1,213}                            & {\color[HTML]{000000} 423}                                & {\color[HTML]{000000} 3}                                & {\color[HTML]{000000} 0}                                  & {\color[HTML]{000000} 21,751}                           & {\color[HTML]{000000} 7,584}                              \\
\textbf{0.01}                                                    & {\color[HTML]{000000} 16}                               & {\color[HTML]{000000} 17}                                 & {\color[HTML]{000000} 222}                              & {\color[HTML]{000000} 369}                                & {\color[HTML]{000000} 4}                                & {\color[HTML]{000000} 0}                                  & {\color[HTML]{000000} 2,927}                            & {\color[HTML]{000000} 6,805}                              \\
\textbf{0.025}                                                   & {\color[HTML]{000000} 9}                                & {\color[HTML]{000000} 9}                                  & {\color[HTML]{000000} 127}                              & {\color[HTML]{000000} 321}                                & {\color[HTML]{000000} 3}                                & {\color[HTML]{000000} 0}                                  & {\color[HTML]{000000} 1,530}                            & {\color[HTML]{000000} 5,614}                              \\
\textbf{0.05}                                                    & {\color[HTML]{000000} 5}                                & {\color[HTML]{000000} 6}                                  & {\color[HTML]{000000} 85}                               & {\color[HTML]{000000} 272}                                & {\color[HTML]{000000} 3}                                & {\color[HTML]{000000} 0}                                  & {\color[HTML]{000000} 1,109}                            & {\color[HTML]{000000} 4,570}                              \\
\textbf{0.1}                                                     & {\color[HTML]{000000} 3}                                & {\color[HTML]{000000} 4}                                  & {\color[HTML]{000000} 54}                               & {\color[HTML]{000000} 199}                                & {\color[HTML]{000000} 3}                                & {\color[HTML]{000000} 0}                                  & {\color[HTML]{000000} 576}                              & {\color[HTML]{000000} 3,056}                             
\end{tabular}
}

%% file: tbl_COL.tex
\resizebox{\columnwidth}{!}{%
\begin{tabular}{cllllllll}
\rowcolor[HTML]{656565} 
\multicolumn{9}{c}{\cellcolor[HTML]{656565}\textbf{Colorado (COL)}}                                                                                                                                                                                                                                                                                                                                                                                                                                                                                  \\
\rowcolor[HTML]{9B9B9B} 
\multicolumn{9}{c}{\cellcolor[HTML]{9B9B9B}\textbf{435,666 states, 1,042,400 edges}}                                                                                                                                                                                                                                                                                                                                                                                                                                                                       \\
\rowcolor[HTML]{C0C0C0} 
\cellcolor[HTML]{C0C0C0}                                         & \multicolumn{2}{c}{\cellcolor[HTML]{C0C0C0}\textbf{avg $n_{sol}$}}                                                  & \multicolumn{2}{c}{\cellcolor[HTML]{C0C0C0}\textbf{avg t}}                                                          & \multicolumn{2}{c}{\cellcolor[HTML]{C0C0C0}\textbf{min t}}                                                          & \multicolumn{2}{c}{\cellcolor[HTML]{C0C0C0}\textbf{max t}}                                                          \\
\rowcolor[HTML]{C0C0C0} 
\multirow{-2}{*}{\cellcolor[HTML]{C0C0C0}\textbf{$\varepsilon$}} & \multicolumn{1}{c}{\cellcolor[HTML]{C0C0C0}\textbf{\is}} & \multicolumn{1}{c}{\cellcolor[HTML]{C0C0C0}\textbf{\boastar}} & \multicolumn{1}{c}{\cellcolor[HTML]{C0C0C0}\textbf{\is}} & \multicolumn{1}{c}{\cellcolor[HTML]{C0C0C0}\textbf{\boastar}} & \multicolumn{1}{c}{\cellcolor[HTML]{C0C0C0}\textbf{\is}} & \multicolumn{1}{c}{\cellcolor[HTML]{C0C0C0}\textbf{\boastar}} & \multicolumn{1}{c}{\cellcolor[HTML]{C0C0C0}\textbf{\is}} & \multicolumn{1}{c}{\cellcolor[HTML]{C0C0C0}\textbf{\boastar}} \\
\textbf{0}                                                       & {\color[HTML]{000000} 318}                              & {\color[HTML]{000000} 318}                                & {\color[HTML]{000000} 3,368}                            & {\color[HTML]{000000} 1,144}                              & {\color[HTML]{000000} 5}                                & {\color[HTML]{000000} 1}                                  & {\color[HTML]{000000} 56,153}                           & {\color[HTML]{000000} 17,348}                             \\
\textbf{0.01}                                                    & {\color[HTML]{000000} 15}                               & {\color[HTML]{000000} 16}                                 & {\color[HTML]{000000} 372}                              & {\color[HTML]{000000} 944}                                & {\color[HTML]{000000} 5}                                & {\color[HTML]{000000} 1}                                  & {\color[HTML]{000000} 3,633}                            & {\color[HTML]{000000} 16,304}                             \\
\textbf{0.025}                                                   & {\color[HTML]{000000} 7}                                & {\color[HTML]{000000} 8}                                  & {\color[HTML]{000000} 192}                              & {\color[HTML]{000000} 768}                                & {\color[HTML]{000000} 5}                                & {\color[HTML]{000000} 1}                                  & {\color[HTML]{000000} 1,690}                            & {\color[HTML]{000000} 15,037}                             \\
\textbf{0.05}                                                    & {\color[HTML]{000000} 4}                                & {\color[HTML]{000000} 5}                                  & {\color[HTML]{000000} 116}                              & {\color[HTML]{000000} 608}                                & {\color[HTML]{000000} 5}                                & {\color[HTML]{000000} 1}                                  & {\color[HTML]{000000} 991}                              & {\color[HTML]{000000} 13,718}                             \\
\textbf{0.1}                                                     & {\color[HTML]{000000} 3}                                & {\color[HTML]{000000} 3}                                  & {\color[HTML]{000000} 69}                               & {\color[HTML]{000000} 470}                                & {\color[HTML]{000000} 4}                                & {\color[HTML]{000000} 1}                                  & {\color[HTML]{000000} 593}                              & {\color[HTML]{000000} 11,977}                            
\end{tabular}
}

%% file: tbl_FLA.tex
\resizebox{\columnwidth}{!}{%
\begin{tabular}{cllllllll}
\rowcolor[HTML]{656565} 
\multicolumn{9}{c}{\cellcolor[HTML]{656565}\textbf{Florida (FL)}}                                                                                                                                                                                                                                                                                                                                                                                                                                                                                  \\
\rowcolor[HTML]{9B9B9B} 
\multicolumn{9}{c}{\cellcolor[HTML]{9B9B9B}\textbf{1,070,376 states, 2,712,798 edges}}                                                                                                                                                                                                                                                                                                                                                                                                                                                                       \\
\rowcolor[HTML]{C0C0C0} 
\cellcolor[HTML]{C0C0C0}                                         & \multicolumn{2}{c}{\cellcolor[HTML]{C0C0C0}\textbf{avg $n_{sol}$}}                                                  & \multicolumn{2}{c}{\cellcolor[HTML]{C0C0C0}\textbf{avg t}}                                                          & \multicolumn{2}{c}{\cellcolor[HTML]{C0C0C0}\textbf{min t}}                                                          & \multicolumn{2}{c}{\cellcolor[HTML]{C0C0C0}\textbf{max t}}                                                          \\
\rowcolor[HTML]{C0C0C0} 
\multirow{-2}{*}{\cellcolor[HTML]{C0C0C0}\textbf{$\varepsilon$}} & \multicolumn{1}{c}{\cellcolor[HTML]{C0C0C0}\textbf{\is}} & \multicolumn{1}{c}{\cellcolor[HTML]{C0C0C0}\textbf{\boastar}} & \multicolumn{1}{c}{\cellcolor[HTML]{C0C0C0}\textbf{\is}} & \multicolumn{1}{c}{\cellcolor[HTML]{C0C0C0}\textbf{\boastar}} & \multicolumn{1}{c}{\cellcolor[HTML]{C0C0C0}\textbf{\is}} & \multicolumn{1}{c}{\cellcolor[HTML]{C0C0C0}\textbf{\boastar}} & \multicolumn{1}{c}{\cellcolor[HTML]{C0C0C0}\textbf{\is}} & \multicolumn{1}{c}{\cellcolor[HTML]{C0C0C0}\textbf{\boastar}} \\
\textbf{0}                                                       & {\color[HTML]{000000} 357}                              & {\color[HTML]{000000} 357}                                & {\color[HTML]{000000} 12,177}                           & {\color[HTML]{000000} 3,545}                              & {\color[HTML]{000000} 12}                               & {\color[HTML]{000000} 3}                                  & {\color[HTML]{000000} 270,450}                          & {\color[HTML]{000000} 68,467}                             \\
\textbf{0.01}                                                    & {\color[HTML]{000000} 12}                               & {\color[HTML]{000000} 13}                                 & {\color[HTML]{000000} 1,000}                            & {\color[HTML]{000000} 3,228}                              & {\color[HTML]{000000} 12}                               & {\color[HTML]{000000} 3}                                  & {\color[HTML]{000000} 17,092}                           & {\color[HTML]{000000} 64,642}                             \\
\textbf{0.025}                                                   & {\color[HTML]{000000} 6}                                & {\color[HTML]{000000} 6}                                  & {\color[HTML]{000000} 479}                              & {\color[HTML]{000000} 2,738}                              & {\color[HTML]{000000} 11}                               & {\color[HTML]{000000} 3}                                  & {\color[HTML]{000000} 8,060}                            & {\color[HTML]{000000} 59,908}                             \\
\textbf{0.05}                                                    & {\color[HTML]{000000} 3}                                & {\color[HTML]{000000} 4}                                  & {\color[HTML]{000000} 263}                              & {\color[HTML]{000000} 1,985}                              & {\color[HTML]{000000} 12}                               & {\color[HTML]{000000} 3}                                  & {\color[HTML]{000000} 3,945}                            & {\color[HTML]{000000} 39,214}                             \\
\textbf{0.1}                                                     & {\color[HTML]{000000} 2}                                & {\color[HTML]{000000} 2}                                  & {\color[HTML]{000000} 144}                              & {\color[HTML]{000000} 1,172}                              & {\color[HTML]{000000} 11}                               & {\color[HTML]{000000} 2}                                  & {\color[HTML]{000000} 1,780}                            & {\color[HTML]{000000} 21,665}                            
\end{tabular}
}

%% file: tbl_NE.tex
\resizebox{\columnwidth}{!}{%
\begin{tabular}{c ll ll ll}
\rowcolor[HTML]{656565} 
\multicolumn{7}{c}{\cellcolor[HTML]{656565}\textbf{North East (NE)}}                                                                                                                                                                                                                                                                                                                                                                                                                                                                                                                                                        \\
\rowcolor[HTML]{9B9B9B} 
\multicolumn{7}{c}{\cellcolor[HTML]{9B9B9B}\textbf{1,524,453 states, 3,897,636 edges}}                                                                                                                                                                                                                                                                                                                                                                                                                                                                                                                                             \\
\rowcolor[HTML]{C0C0C0} 
\cellcolor[HTML]{C0C0C0}                                         & \multicolumn{2}{c}{\cellcolor[HTML]{C0C0C0}\textbf{avg t}}                                                                                                                         & \multicolumn{2}{c}{\cellcolor[HTML]{C0C0C0}\textbf{min t}}                                                                                                                         & \multicolumn{2}{c}{\cellcolor[HTML]{C0C0C0}\textbf{max t}}                                                                                                                         \\
\rowcolor[HTML]{C0C0C0} 
\multirow{-2}{*}{\cellcolor[HTML]{C0C0C0}\textbf{$\varepsilon$}} & \multicolumn{1}{c}{\cellcolor[HTML]{C0C0C0}\textbf{PP-A*}} & \multicolumn{1}{c}{\cellcolor[HTML]{C0C0C0}\textbf{BOA*}} & \multicolumn{1}{c}{\cellcolor[HTML]{C0C0C0}\textbf{PP-A*}} & \multicolumn{1}{c}{\cellcolor[HTML]{C0C0C0}\textbf{BOA*}} &  \multicolumn{1}{c}{\cellcolor[HTML]{C0C0C0}\textbf{PP-A*}} & \multicolumn{1}{c}{\cellcolor[HTML]{C0C0C0}\textbf{BOA*}} 
\\
\textbf{0}                                                       & {\color[HTML]{000000} 192.6}                            & {\color[HTML]{000000} 59.5}                              & {\color[HTML]{000000} 0.04}                             & {\color[HTML]{000000} 0.02}                               & {\color[HTML]{000000} 2,4189.9}
& {\color[HTML]{000000} 592.6}                             \\
\textbf{0.01}                                                    & {\color[HTML]{000000} 13.1}                             & {\color[HTML]{000000} 68.3}                              & {\color[HTML]{000000} 0.03}                             & {\color[HTML]{000000} 0.01}                               & {\color[HTML]{000000} 111.6}                            & {\color[HTML]{000000} 600.9}                             \\
\textbf{0.025}                                                   & {\color[HTML]{000000} 5.6}                             & {\color[HTML]{000000} 57.3}                              & {\color[HTML]{000000} 0.02}                             & {\color[HTML]{000000} 0.01}                               & {\color[HTML]{000000} 46.9}                             & {\color[HTML]{000000} 510.9}                             \\
\textbf{0.05}                                                    & {\color[HTML]{000000} 2.7}                              & {\color[HTML]{000000} 40.8}                              & {\color[HTML]{000000} 0.02}                             & {\color[HTML]{000000} 0.01}                               & {\color[HTML]{000000} 22.6}                             & {\color[HTML]{000000} 345.1}                             \\
\textbf{0.1}                                                     & {\color[HTML]{000000} 1.3}                              & {\color[HTML]{000000} 25.8}                               & {\color[HTML]{000000} 0.02}                             & {\color[HTML]{000000} 0.01}                               & {\color[HTML]{000000} 9.0}                             & {\color[HTML]{000000} 229.8}                              \end{tabular}
}